\documentclass[letterpaper, 12pt]{amsart}
	\usepackage{amsmath, amsfonts, amsthm, mathrsfs, amssymb, graphicx, subfigure}
	\usepackage[margin = 1in]{geometry}

	\usepackage{fancyhdr}

	\usepackage{color}
	\definecolor{myblue}{rgb}{0.3, 0.0, 0.85}
	\definecolor{myviolet}{rgb}{0.5, 0.0, 0.5}

	\usepackage{hyperref}
	\hypersetup{colorlinks = true, linkcolor = myblue, citecolor = myviolet}

	\theoremstyle{plain}
	\newtheorem{defn}{Definition}[section]
	\newtheorem{prop}[defn]{Proposition}
	\newtheorem{lem}[defn]{Lemma}
	\newtheorem{thm}[defn]{Theorem}
	
	\newtheorem{rem}[defn]{Remark}

	\title{Nonlinear Sturm Oscillation: from the interval to a star}
	\author{Ram Band}
	\address{Amado Building, Technion Math.\ Dept., Haifa, 32000, Israel}
	\email{ramband@tx.technion.ac.il}
	\author{August J. Krueger}
	\address{Amado Building, Technion Math.\ Dept., Haifa, 32000, Israel}
	\email{ajkrueger@tx.technion.ac.il}
	\date{May. 16th, 2017}

\begin{document}

	\begin{abstract}
	The Sturm oscillation property, i.e. that the $n$-th eigenfunction of a Sturm-Liouville operator on an interval has $n -1$ zeros (nodes), has been well studied. This result is known to hold when the interval is replaced by a metric (quantum) tree graph. We prove that the solutions of the real stationary nonlinear Schr\"odinger equation on an interval satisfy a nonlinear version of the Sturm oscillation property. However, we show that unlike the linear theory, the nonlinear version of the Sturm oscillation breaks down already for a star graph. We point out conditions under which this violation can be assured.
	\end{abstract}

	\maketitle


	\section{Introduction}
	
	The linear theory of Sturm-Liouville operators, and the associated oscillation theorems, that began in \cite{Sturm-1, Sturm-2} has lead to an extensive and robust field of ideas and results. See, for example, \cite{Sturm_compilation} for a broad review of the classical and modern theory. Put simply, Sturm oscillation theorem states that if the eigenvalues of an operator are indexed increasingly by $\mathbb N$, the $n$-th eigenfunction has $n -1$ interior zeros. The theory of Sturm oscillation may be extended in many different directions, one of which is to consider differential operators on collections of line segments joined at their endpoints with suitable matching conditions. Theses networks, or graphs, are called tree graphs, if they admit no closed cycles. For the cases where the differential operator is the Laplacian with Robin matching conditions, the Sturm oscillation property for a tree graph has been established, see e.g. recent results in \cite{Be07, Tree-1, Tree-2} and review in \cite{Tree-3}.

	We consider the following generalization of oscillation theory: to nonlinear differential equations on line segments and star graphs. This extension immediately prohibits a direct appeal to linear spectral theory and therefore new definitions must be given not only for the operators in question but also of a suitable notion of nonlinear Sturm oscillation.

	In the linear theory, one may rescale eigenfunctions with impunity. The nonlinear theory, however, lacks such a trivial scale factor. Therefore, in order to characterize all stationary solutions of the nonlinear Schr\"odinger equation one needs to introduce an additional parameter. Such a parameter is usually taken to be some norm of the solution and we take it here to be the $L^\infty$ norm. In the two-dimensional space which is parametrized by the spectral parameter and the norm, one may represent families of stationary solutions as simple curves, which we call spectral curves. We prove a kind of nonlinear Sturm oscillation property for the interval, where the spectral curves may be indexed and each solution lying on the $n$-th curve has $n -1$ interior zeros (Theorem \ref{thm_1}). Following this, the nonlinear Schr\"odinger equation on a star graph is considered. The full nonlinear spectrum for the star graph is beyond the scope of this paper, however local properties of spectral curves are explored. We prove that the nodal count is not constant in general along each such spectral curve (Theorem \ref{thm_2}) and show how to construct star graphs and solutions for which such a nodal count change occurs. Therefore, in distinction to the linear theory, the analogous form of nonlinear Sturm oscillation property does not generically hold already for the simplest tree graphs.

	We refer the interested reader to the reviews \cite{BeKu13, GnSm06} on the linear theory of metric (quantum) graphs. The nonlinear Schr\"odinger equation on metric graphs is addressed in many recent works. We mention here only those works which deal with stationary solutions and are closer in spirit to those discussed in the current paper. With regard to stationary solutions of the nonlinear Schr\"odinger equation we note the study of general scattering in \cite{GnSm11, GnScSm13}, of general solutions on star graphs in \cite{Ad-1, Ad-2, Ad-3}, and of bifurcation and stability properties of solutions on various graphs in \cite{Pe-2, Pe-1, Pe-3}. Finally, a framework to aid in the solving of the stationary nonlinear Schr\"odinger equation on metric graphs was presented recently in \cite{GnWa16, GnWa16-1}.

	\section{Main results}
	
	\subsection{Nonlinear Schr\"odinger equation on an interval}
	
	\begin{defn}
	Let the \emph{real, stationary, nonlinear Schr\"odinger equation} on an interval of length $0 < l \in \mathbb R$ be given by
	\begin{align}\label{line}
	\mu\phi = -\partial^2_x\phi -(\sigma +1)\nu\phi^{2\sigma +1},\quad \mu\in \mathbb R,\ 0 \ne \nu \in \mathbb R,\ \sigma \in \mathbb N,\ \phi \in C^2([0, l], \mathbb R)
	\end{align}
	and subject to boundary conditions that can be either of Dirichlet type, $\phi(x_j) = 0$, or of Neumann type, $\partial_x\phi(x_j) = 0$, where $j = 1, 2$, $x_1 = 0$, $x_2 = l$.
	
	One may analogously define \eqref{line} on a ray, e.g. $\phi \in C^2([0, \infty), \mathbb R)$, by eliminating the second boundary condition, and alternatively on a line, e.g. $\phi \in C^2(\mathbb R, \mathbb R)$, by eliminating both boundary conditions.
	\end{defn}
	
	It is easiest to classify the solutions of \eqref{line} on an interval and ray by first considering those on a line. Furthermore, as stated the derivation leading to \eqref{energy}, integrating \eqref{line} yields
	\begin{align}
	h &= (\partial_x\phi)^2 +\mu\phi^2 +\nu\phi^{2(\sigma +1)},\quad h \in \mathbb R,
	\end{align}
	which can be interpreted as a Hamiltonian energy conservation constraint. Here the constant $h$ takes the role as the energy, of a particle moving on a line in the time coordinate $x$, represented by $\phi(x)$. One may therefore visualize the solutions by considering the effective potential energy $h_\mathrm p(\phi) = \mu\phi^2 +\nu\phi^{2(\sigma +1)}$. We will distinguish between four special cases:

	\begin{figure}[t]
	\centering
	\subfigure[Case I, $(\nu > 0, \mu \ge 0)$.]{\includegraphics[width=.47\textwidth]{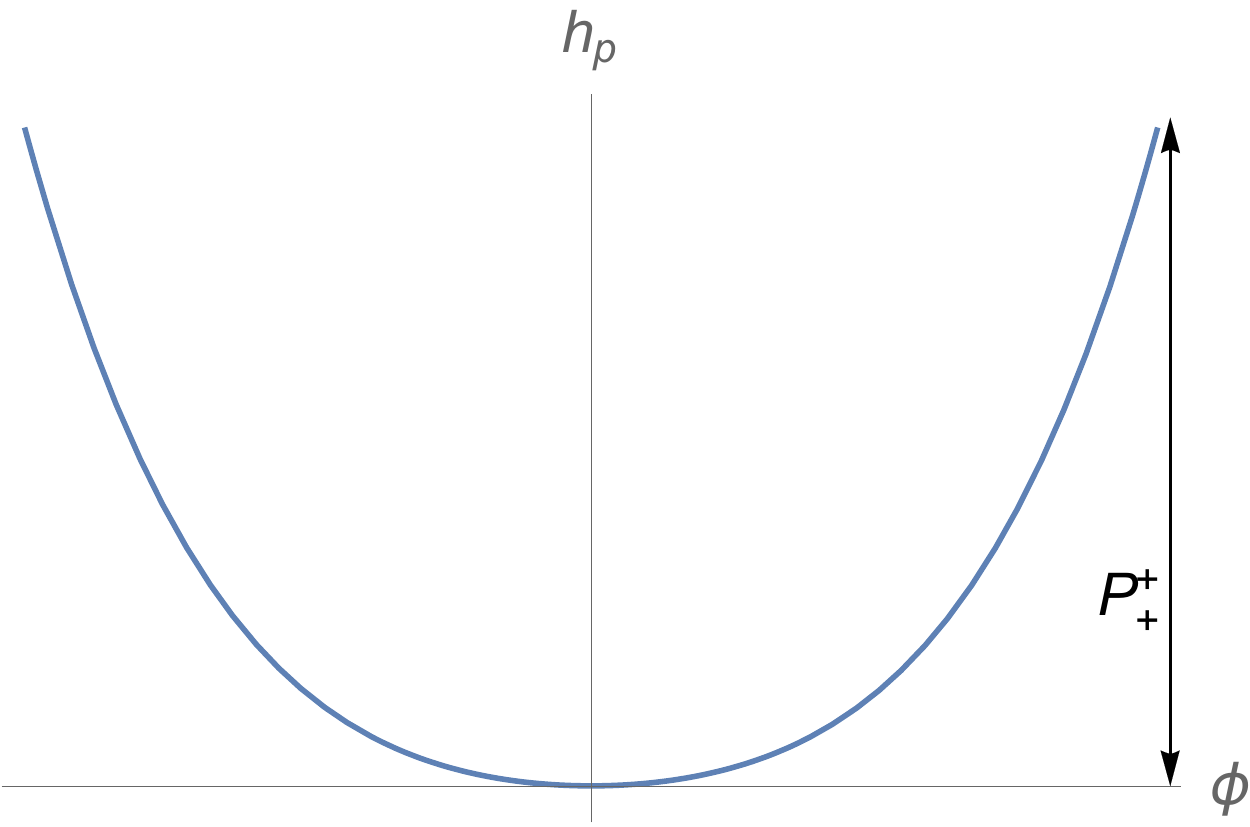}}
	\hfill
	\subfigure[Case II, $(\nu > 0, \mu < 0)$.]{\includegraphics[width=.47\textwidth]{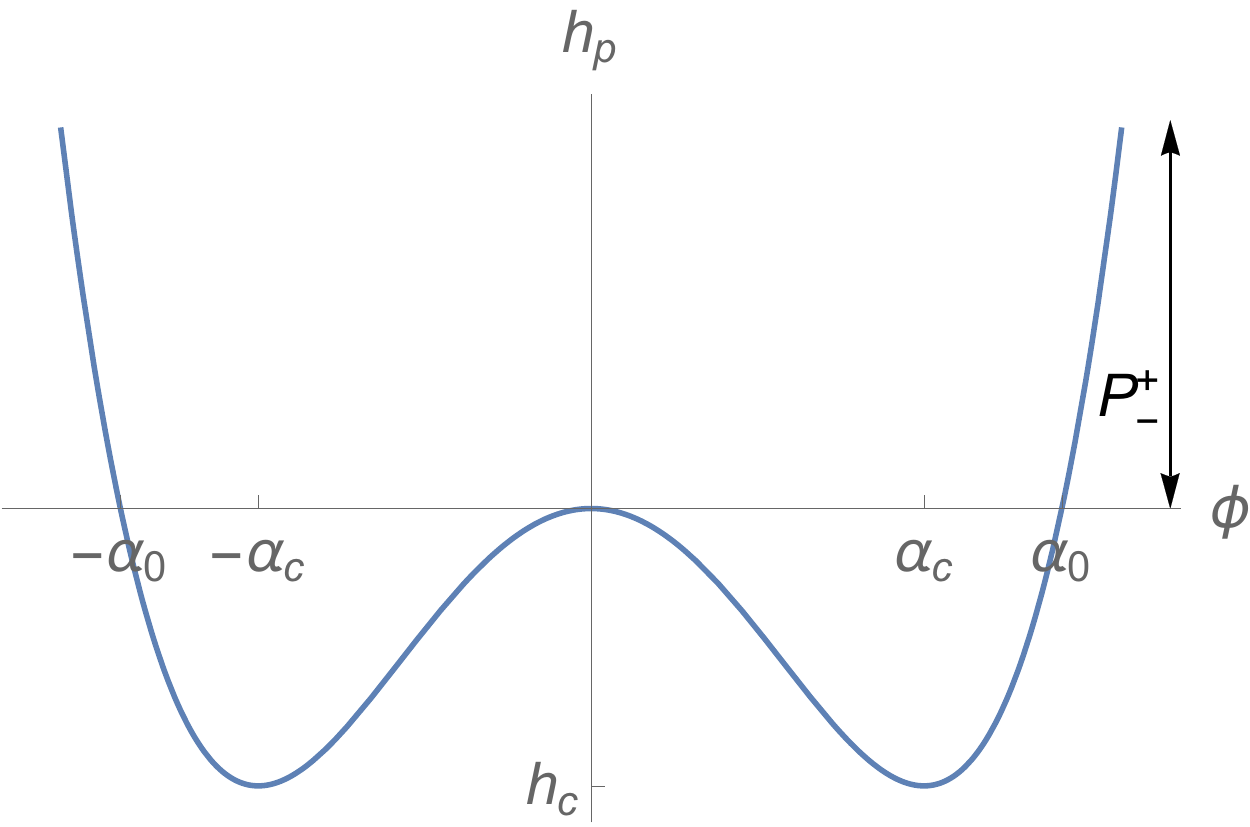}}
	\caption{$\nu > 0$.}\label{case_1_2}
	\end{figure}
	
	\begin{figure}[t]
	\centering
	\subfigure[Case III, $(\nu < 0, \mu \ge 0)$.]{\includegraphics[width=.47\textwidth]{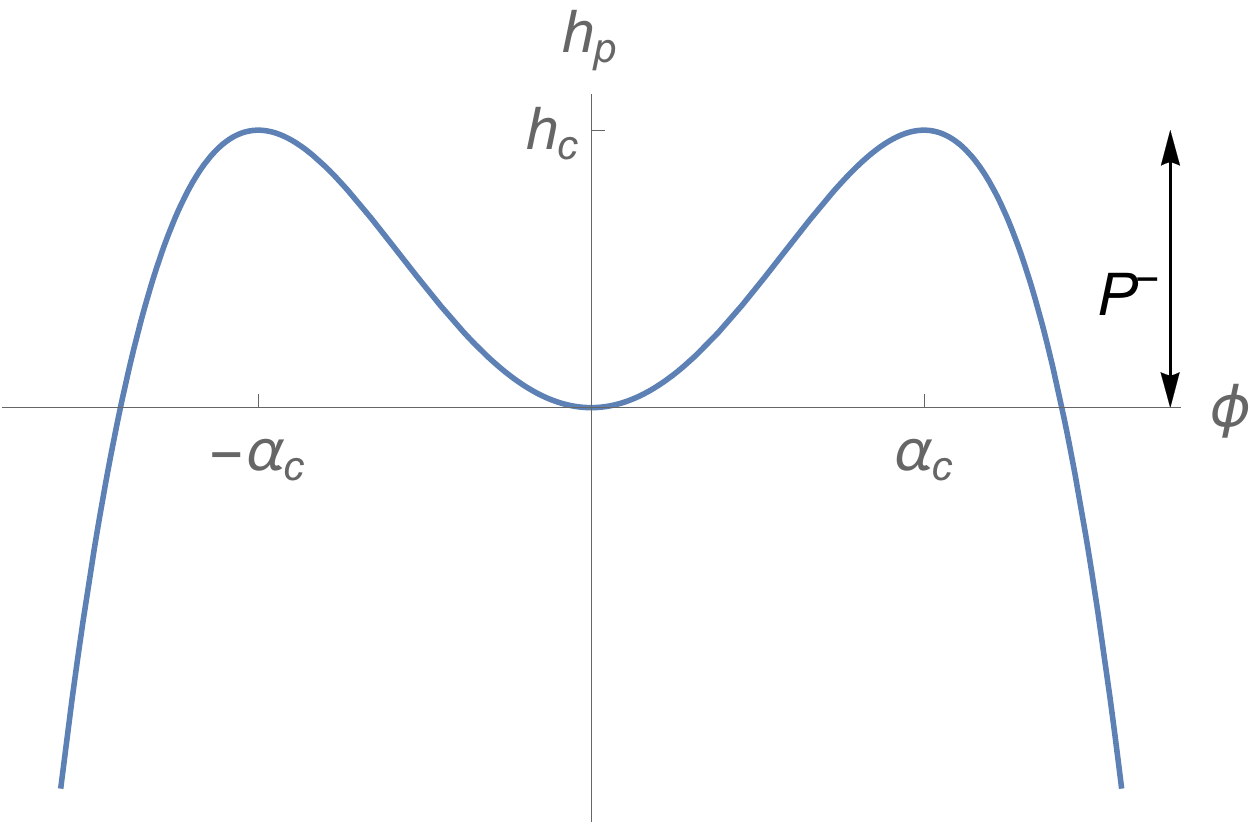}}
	\hfill
	\subfigure[Case IV, $(\nu < 0, \mu < 0)$.]{\includegraphics[width=.47\textwidth]{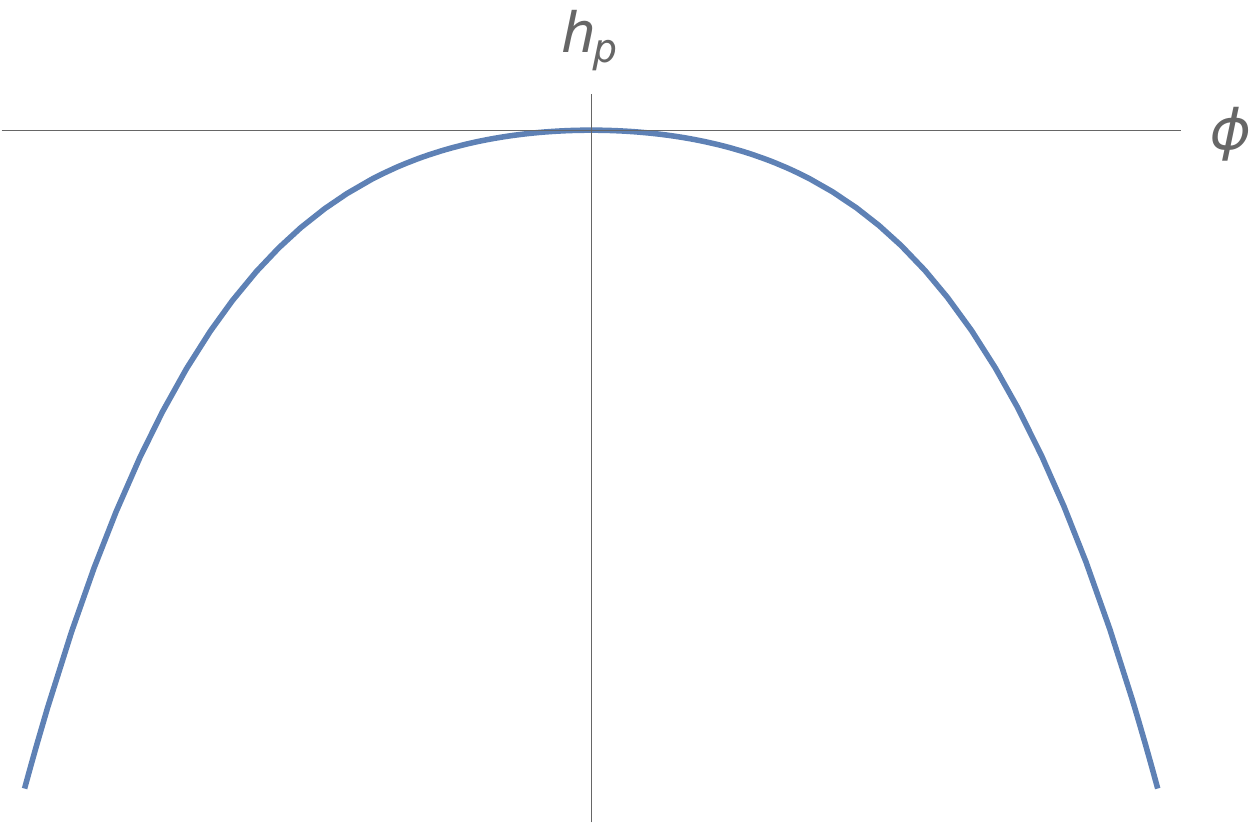}}
	\caption{$\nu < 0$.}\label{case_3_4}
	\end{figure}	
	
	\begin{itemize}
	\item Case I, $(\nu > 0, \mu \ge 0)$. Shown in Figure \ref{case_1_2} (a). $h_\mathrm p(\phi)$ has a global minimum at $\phi = 0$ and $\lim_{\phi \to \pm\infty}h_\mathrm p(\phi) = \infty$.\\
	\item Case II, $(\nu > 0, \mu < 0)$. Shown in Figure \ref{case_1_2} (b). $h_\mathrm p(\phi)$ has a local maximum at $\phi = 0$, global minima at $\phi = \pm|\mu/(\sigma +1)\nu|^{1/2\sigma}$, and $\lim_{\phi \to \pm\infty}h_\mathrm p(\phi) = \infty$.\\
	\item Case III, $(\nu < 0, \mu \ge 0)$. Shown in Figure \ref{case_3_4} (a). $h_\mathrm p(\phi)$ has a local minimum at $\phi = 0$, global maxima at $\phi = \pm|\mu/(\sigma +1)\nu|^{1/2\sigma}$, and $\lim_{\phi \to \pm\infty}h_\mathrm p(\phi) = -\infty$.\\
	\item Case IV, $(\nu < 0, \mu < 0)$. Shown in Figure \ref{case_3_4} (b). $h_\mathrm p(\phi)$ has a global maximum at $\phi = 0$ and $\lim_{\phi \to \pm\infty}h_\mathrm p(\phi) = -\infty$.
	\end{itemize}
	\begin{defn}
	We associate with each bounded solution of \eqref{line} on a line a parameter $\alpha = ||\phi||_\infty$. We define the following distinguished subsets of $\mathbb R^2$. For $\nu > 0$:
	\begin{align}
	P^+_+ := \{(\mu, \alpha) \in \mathbb R^2 : \mu > 0, \alpha > 0\},\quad P^+_- := \{(\mu, \alpha) \in \mathbb R^2 : \mu \le 0, \alpha > \alpha_0\},
	\end{align}
	and for $\nu < 0$:
	\begin{align}
	P^- &:= \{(\mu, \alpha) \in \mathbb R^2 : \mu > 0, 0 < \alpha < \alpha_\mathrm c\},
	\end{align}
	where $\alpha_0 := |\mu/\nu|^{1/2\sigma}$ and $\alpha_\mathrm{c} := |\mu/(\sigma +1)\nu|^{1/2\sigma}$. Furthermore we take
	\begin{align}
	P^+ := P^+_+ \cup P^+_-,\quad P := P^+ \cup P^-.
	\end{align}
	\end{defn}
	
	We would like to study the properties of solutions that oscillate symmetrically through zero, as these are most similar to the solutions of the analogous linear system, i.e. $\nu = 0$. Let $\nu \ne 0$ and fix a Dirichlet or Neumann boundary condition at each endpoint of an interval, we denote
	\begin{align}
	\Phi_\mathrm{int} := \{\phi \ne 0\text{ solves \eqref{line} on an interval, such that }\phi\text{ attains at least one zero}\}.
	\end{align}
	The next lemma parametrizes $\Phi_\mathrm{int}$ by $P$.
	
	\begin{lem}\label{lem_int}
	Fix an interval of length $l > 0$ and $\nu \ne 0$. The following holds.
	\begin{enumerate}
	\item\label{lem_int_1} Given $\phi \in \Phi_\mathrm{int}$, there is a unique value of $\mu \in \mathbb R$ such that $\phi$ is a solution of \eqref{line} on this interval with the given values $\mu, \nu$. This allows one to define the map
	\begin{align}
	\Lambda_\mathrm{int}: \Phi_\mathrm{int} \to P,\quad \Lambda_\mathrm{int}: \phi \mapsto (\mu, ||\phi||_\infty).
	\end{align}
	\item\label{lem_int_2} $\Lambda_\mathrm{int}$ is two to one since $\Lambda_\mathrm{int}(\phi_1) = \Lambda_\mathrm{int}(\phi_2)$ if and only if $\phi_1 = \zeta\phi_2$, where $\zeta = \pm1$.
	\end{enumerate}
	\end{lem}	
	The above lemma allows one to parametrize solutions of \eqref{line} with points $(\mu, \alpha) \in P$, which we write as $\phi = \phi_{(\mu, \alpha)}$. Those solutions lie on curves in the $(\mu, \alpha)$ half plane, as is demonstrated in Figure \ref{curves}. This is stated in an exact manner in the next theorem, which also establishes a nonlinear form of the Sturm oscillation property for the interval.
	
	\begin{figure}[t]
	\centering
	\subfigure[$\nu > 0$]{\includegraphics[width=.47\textwidth]{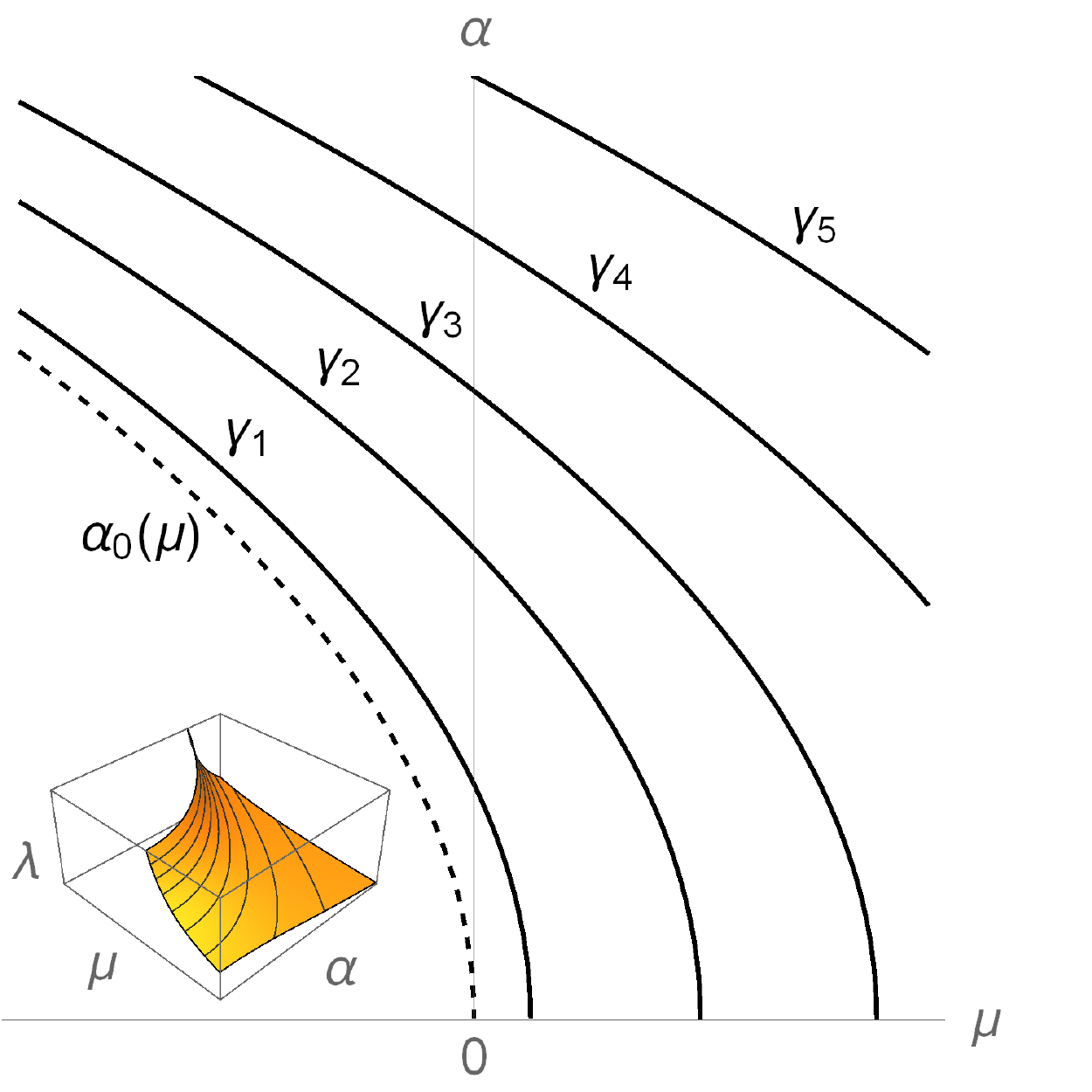}}
	\hfill
	\subfigure[$\nu < 0$]{\includegraphics[width=.47\textwidth]{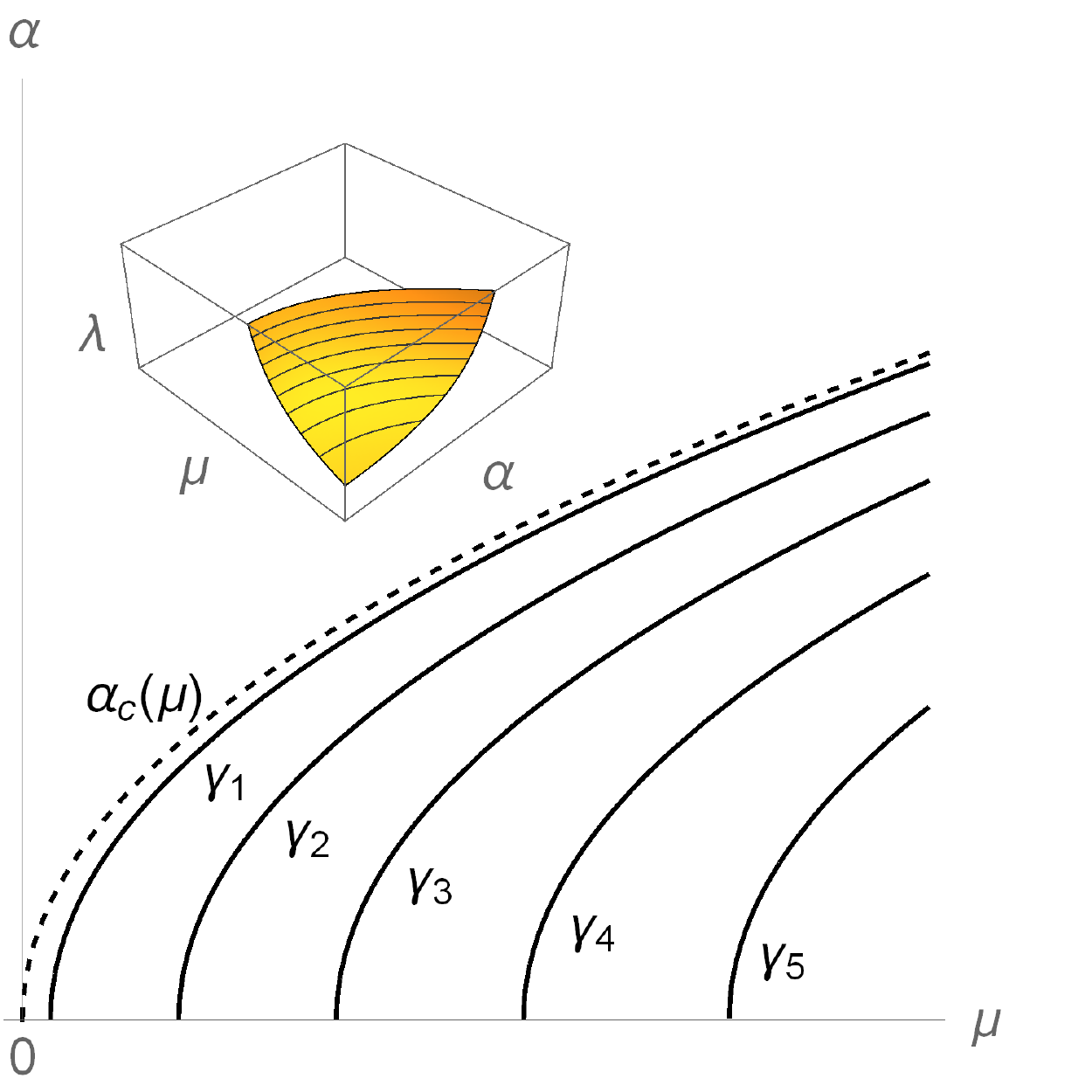}}
	\caption{A sketch of nonlinear spectral curves on an interval. The inset figures illustrate that the curves are level sets of the wavelength $\lambda_{(\mu, \alpha)}$.}\label{curves}
	\end{figure}
	
	\begin{thm}\label{thm_1}
	Fix an interval of length $l > 0$ and $\nu \ne 0$. The following holds
	\begin{enumerate}
	\item\label{thm_1_part_1} $\Lambda_\mathrm{int}\left(\Phi\right) = \bigsqcup_{n \in \mathscr N}\gamma_n$, where this is a disjoint union, each $\gamma_n$ is a connected, non self intersecting curve, and $\mathscr N = \mathbb N$ if at least one boundary condition is Dirichlet and $\mathscr N = \mathbb N\setminus\{1\}$ if both are Neumann.
	\item\label{thm_1_part_2} If $\Lambda_\mathrm{int}(\phi) \in \gamma_n$, then $\phi$ has $n -1$ interior zeros.
	\item\label{thm_1_part_3} Let $\alpha \in \mathbb R$ be fixed. Each $\gamma_n$ intersects the line $\{(\mu, \alpha)\}_{\mu \in \mathbb R}$ only once. Furthermore these intersection points occur for $\mu = \mu_n$, where the $\mu_n$ are monotonically strictly increasing in $n$.
	\item\label{thm_1_part_4} $\lim_{\alpha \to 0}\gamma_n = (\mu^\mathrm{lin}_n, 0)$, where $\mu^\mathrm{lin}_n$ is the $n$-th eigenvalue of the linear problem on an interval.
	\end{enumerate}
	\end{thm}
	
	The first two parts of the theorem above show that the solutions of the nonlinear Schr\"odinger equation on an interval are naturally arranged in a sequential order and that all the solutions corresponding to the $n$-th set (curve) possess $n -1$ internal nodal points. We treat this as the nonlinear analogue of the Sturm oscillation property. The third part shows that once the solution norm, $\alpha$, is fixed one obtains a discrete spectrum of solutions which obey Sturm oscillation. The fourth part connects this with the linear spectrum which is obtained as $\alpha \to 0$.
	
	\subsection{Nonlinear Schr\"odinger equation on a star graph}
	
	\begin{defn}\label{def_star}
	Consider a set of $d > 2$ intervals with edge lengths $0 < l_j \in \mathbb R$, $j = 1, \ldots, d$. Join one endpoint of each of these intervals, hereafter termed \emph{edges}, at a single point, hereafter termed the \emph{central vertex}, and denote the resulting set a \emph{star graph}, $\Gamma$, of degree $d$, whose endpoints are hereafter termed boundary vertices. We endow $\Gamma$ with a fixed coordinate system where $x_j \in [0, l_j]$ is a coordinate on edge $j$ such that $x_j = 0$ at the central vertex along each edge and $x_j = l_j$ at each boundary vertex. For any function $\phi: \Gamma \to \mathbb R$, we take $\phi_j$ be its restriction to edge $j$.
	
	Let the real, stationary, nonlinear Schr\"odinger equation on a degree $d > 2$ star graph $\Gamma$ be given by
	\begin{align}\label{star}
	\mu\phi_j = -\partial^2_x\phi_j -(\sigma +1)\nu\phi_j^{2\sigma +1}, \quad\mu \in \mathbb R,\ 0 \ne \nu \in \mathbb R,\ \sigma \in \mathbb N,\ \phi_j \in C^2([0, l_j], \mathbb R),
	\end{align}
	for $j = 1, \ldots, d$, with a Neumann condition at the central vertex
	\begin{align}
	&\phi_1(0) = \phi_j(0),\quad j = 2, \ldots, d,\\
	&\sum_{j = 1}^d\partial_x\phi_j(0) = 0,
	\end{align}
	and a Dirichlet or Neumann condition at each boundary vertex, $x_j = l_j$.
	\end{defn}
	
	Let $\nu \ne 0$ and fix a Dirichlet or Neumann boundary condition at each endpoint of a graph. We denote
	\begin{align}
	\Phi_\mathrm{star} := \{\phi \ne 0\text{ solves \eqref{star}, such that }\phi\text{ attains at least one zero}\},
	\end{align}
	and parametrize $\Phi_\mathrm{star}$, in a similar manner as was done for the solutions on the interval.
	
	\begin{defn}
	Let $Q \subset \mathbb R^{d +1}$ be the space of points $q = (\mu, \alpha_1, \ldots, \alpha_d)$ such that $(\mu, \alpha_j) \in P$, for all $j = 1, \ldots, d$ and fixed $\nu \ne 0$.
	\end{defn}
	
	\begin{lem}\label{lem_star}
	Fix a star graph with edges of length $l_j > 0$, $j = 1, \ldots, d$, and $\nu \ne 0$. Given $\phi \in \Phi_\mathrm{star}$, there is a unique value of $\mu \in \mathbb R$ such that $\phi$ is a solution of \eqref{star} on this graph with the given values $\mu, \nu$. This allows us to define the map
	\begin{align}
	\Lambda_\mathrm{star}: \Phi_\mathrm{star} \to Q,\quad \Lambda_\mathrm{star}: \phi \mapsto (\mu, ||\phi_1||_\infty, \ldots, ||\phi_d||_\infty),
	\end{align}
	for which $\Lambda_\mathrm{star}(\phi^{(1)}) = \Lambda_\mathrm{star}(\phi^{(2)})$ implies $\phi^{(1)}_j = \zeta_j\phi^{(2)}_j$, where $\zeta_j = \pm1$ for all $j$.
	\end{lem}
	The above lemma allows one to parametrize solutions of \eqref{star} with points $q \in Q$, which we write as $\phi = \phi_{(q)}$.
	
	\begin{defn}
	If there exists a continuous map $\gamma: \mathbb R \to Q$ such that $\phi_{(\gamma(\tau))} \in \Phi_\mathrm{star}$ for all $\tau$ then we term $\gamma$ a \emph{local spectral curve}.
	\end{defn}
	
	For the interval, we managed to decompose $\Phi_\mathrm{int}$ as a union of spectral curves, such that the nodal count of the solution is fixed along each curve. For the star we do not address the global structure of spectral curves but rather show that locally the nodal count may change along the spectral curves, thereby preventing a nonlinear Sturm oscillation property from being satisfied on the star.
	
	\begin{thm}\label{thm_2}
	
	Let $\phi_{(q_*)} \in \Phi_\mathrm{star}$ for some $q_* \in Q$. If $\phi_{(q_*)}$ vanishes at the central vertex then:
	
	\begin{enumerate}
	
	\item\label{thm_2_part_1} There exists a local spectral curve $\gamma: \mathbb R \to Q$ which passes through $q_* \in Q$.
	
	\item\label{thm_2_part_2} For all $q \in \gamma$ sufficiently close to $q_* \in \gamma$ we have that the change in nodal count between the solutions at $q$ and $q_*$ is given by
	\begin{align}\label{nodal_count_change}
	Z(\phi_{(q)}) -Z(\phi_{(q_*)}) = \mathrm{sgn}^2(\phi_{(q)})\left[-1 +2^{-1}d -2^{-1}\mathrm{sgn}(\phi_{(q)})\sum_{j = 1}^d\mathrm{sgn}(\partial_x\phi_{(q_*), j})\right]\downharpoonright_{x = 0},
	\end{align}
	where $Z(\phi) \in \mathbb N$ is the number of zeros (nodes) of $\phi$ in the interior of $\Gamma$.
	
	\item\label{thm_2_part_3} If there are only Dirichlet conditions on exterior vertices, $\nu > 0$, and the edge lengths $\ell_1, \ldots, \ell_d$ satisfy $\sum_{j = 1}^d \zeta_j(n_j/l_j)^{1 +1/\sigma} = 0$ for some set of $\zeta_j \in \{-1, 1\}$ and $n_j \in \mathbb N$ for all $j$, then there exists a $\phi_{(q_*)} \in \Phi_\mathrm{star}$ with $\mu_* = 0$, $\phi_{(q_*)}(0) = 0$, $\zeta_j = \mathrm{sgn}(\partial_x\phi_{(q_*), j}(0))$, and interior nodal count $Z(\phi_{(q_*)}) = 1 -d +\sum_{j = 1}^dn_j$.
	
	\item\label{thm_2_part_4} In addition to assumptions in \eqref{thm_2_part_3}, we further have that if $\sum_{j = 1}^d\zeta_j(n_j/l_j)^{-1 +1/\sigma} \ne 0$, then through this $q_*$ passes a local spectral curve $\gamma$ such that for all $q_+, q_- \in \gamma$ sufficiently close to $q_*$, where $\mu_+ > 0 > \mu_-$, one has the interior nodal count change $|Z(\phi_{(q_+)}) -Z(\phi_{(q_-)})| = |\sum_{j = 1}^d\zeta_j|$.
	
	\end{enumerate}
	
	\end{thm}
	
	\begin{figure}[t]
	\centering
	\subfigure[4 interior zeros]{\includegraphics[width=.32\textwidth]{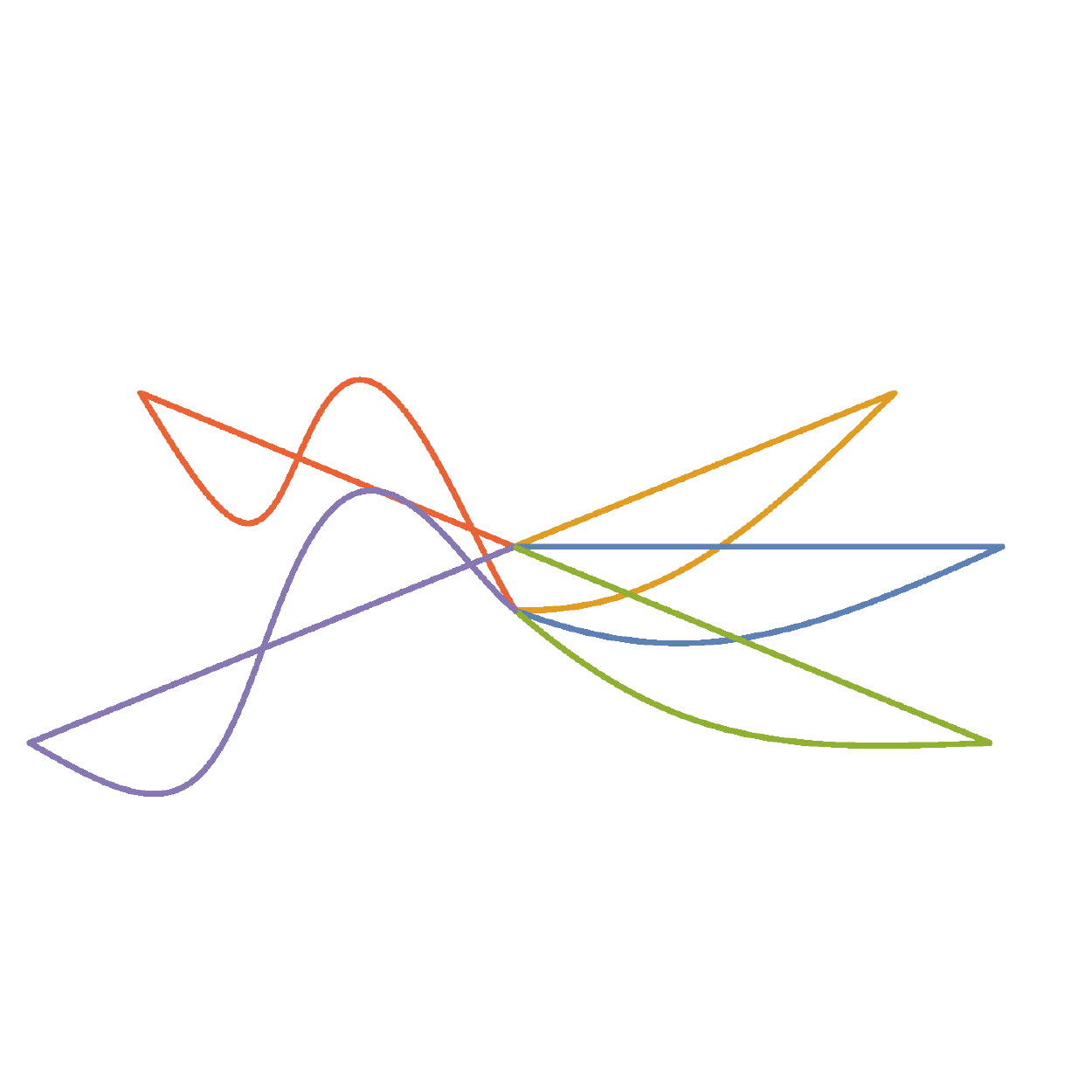}}
	\subfigure[3 interior zeros]{\includegraphics[width=.32\textwidth]{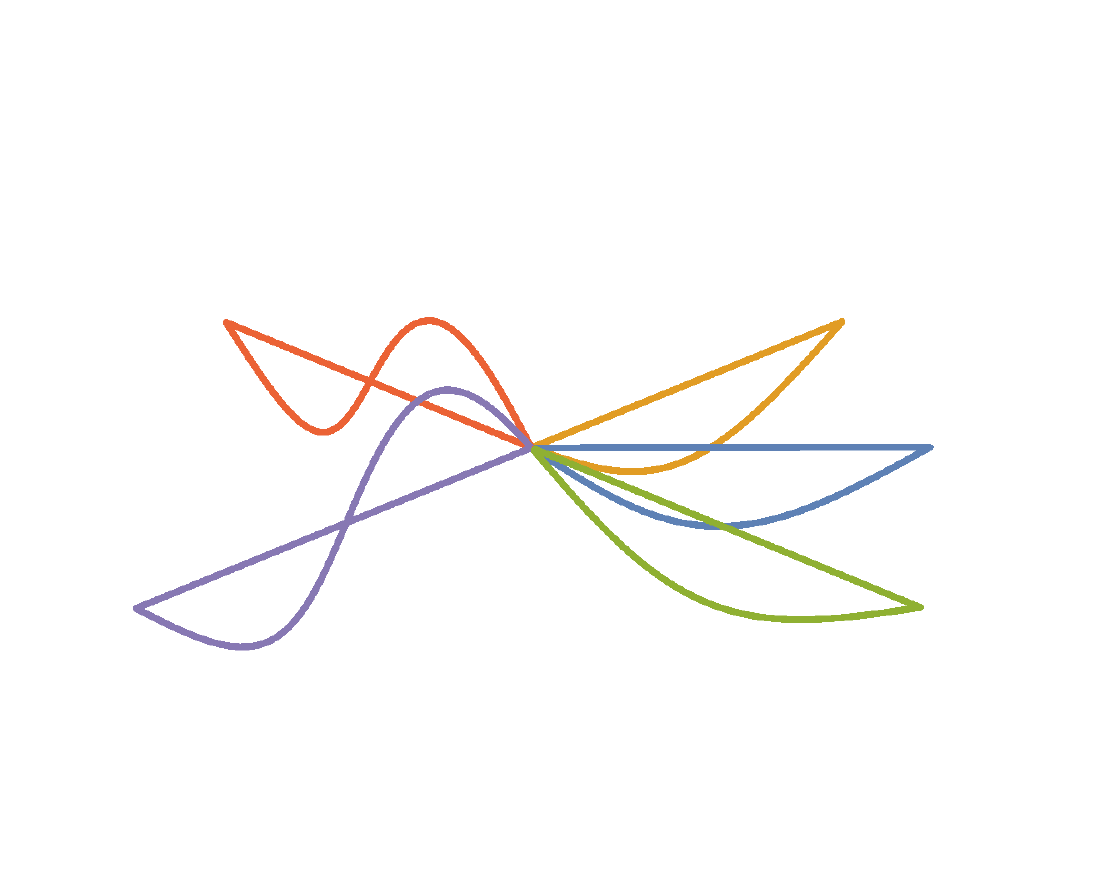}}
	\subfigure[5 interior zeros]{\includegraphics[width=.32\textwidth]{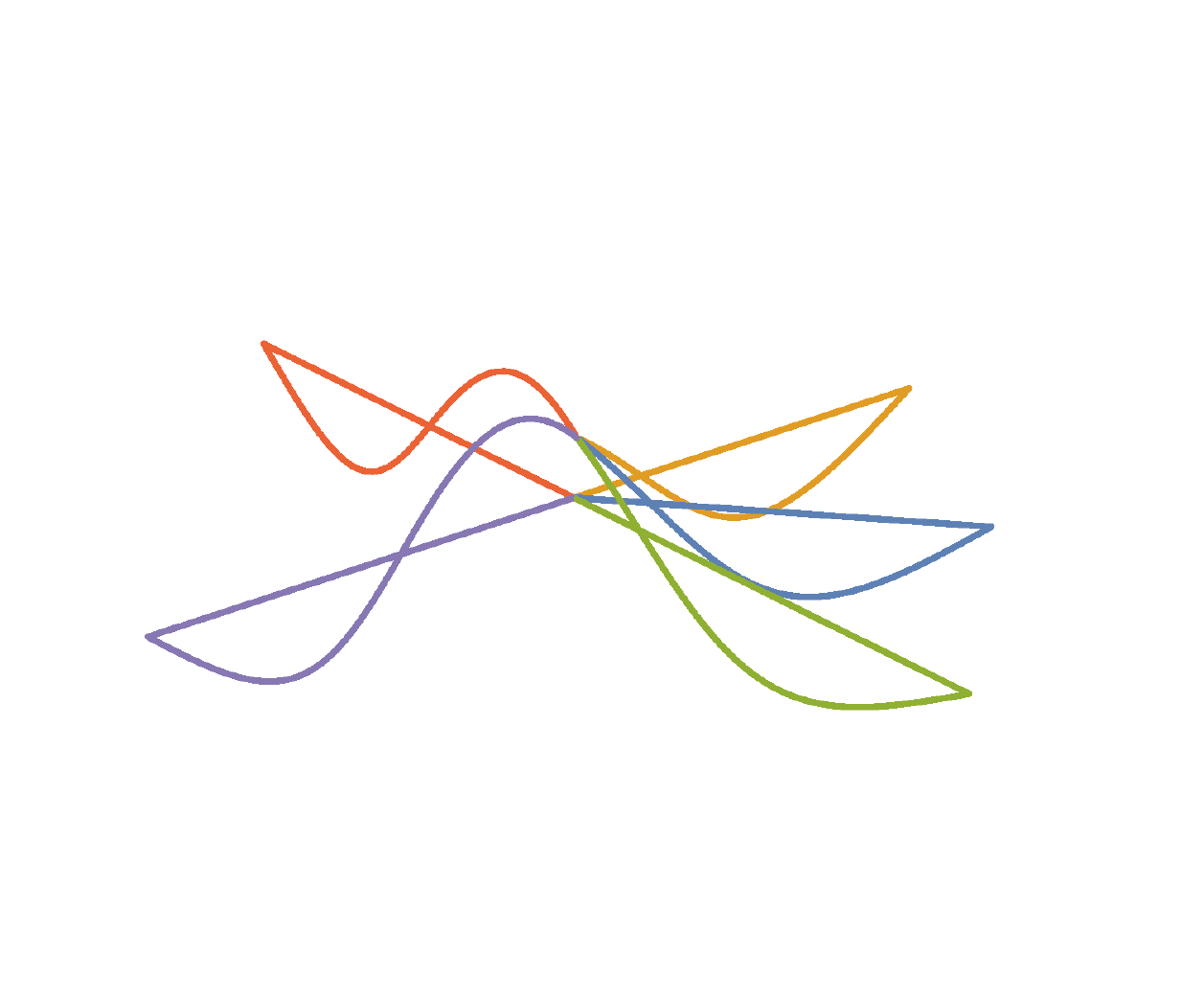}}
	\caption{Sketch of nodal count change along a local spectral curve}\label{nodal_count}
	\end{figure}
	
	\begin{rem}
	Parts \eqref{thm_2_part_3} and \eqref{thm_2_part_4} of Theorem \ref{thm_2} are actually slightly more general and also apply to a star with either a Dirichlet or Neumann on each exterior vertex. The statements would be modified as follows. In \eqref{thm_2_part_3}, we change the condition $\sum_{j = 1}^d \zeta_j(n_j/l_j)^{1 +1/\sigma} = 0$ in a way that each term corresponding to an edge with a Neumann condition becomes $\zeta_j[(n_j -1/2)/l_j]^{1 +1/\sigma}$. A similar change is done for the condition of \eqref{thm_2_part_4}.
	\end{rem}
	
	The theorem above demonstrates that for a star graph we cannot obtain a Sturm oscillation property similar to the one we got for the interval. In order for such property to hold, we need to have that the nodal count is constant along spectral curves. Part \eqref{thm_2_part_1} of the theorem shows the existence of such a local curve, at least locally. Then Part \eqref{thm_2_part_2} of the theorem shows what is the nodal count change between a solution vanishing at the central vertex and a neighboring solution. Finally, the last two parts of the theorem show how to construct neighboring solutions which exhibit a nodal count change. We note that in such a construction, the nodal count change differs than zero for all star graphs with odd number of edges, an example of which is illustrated in Figure \ref{nodal_count}.
	
	The paper is structured as follows. Sections \ref{prelim} and \ref{wavelength} provide some work that is required for the proof of Theorems \ref{thm_1} and \ref{thm_2}. Section \ref{thm_1_proof} presents the proof of Theorem \ref{thm_1} and section \ref{thm_2_proof} presents the proof of Theorem \ref{thm_2}.

	\section{Preliminaries}\label{prelim}
	
	Due to the importance, for general theory as well as applications, of the fact that solutions of the standard stationary nonlinear Scr\"odinger equation are complex valued, we choose to first couch the real stationary solutions in the context of the larger theory of complex stationary solutions. To this end we introduce a convenient, if nonstandard, means of coordinate decomposition.

	We denote by \emph{extended polar coordinates} the pair $(\phi,\theta)$ where $\phi \in \mathbb R$ and $0 \le \theta < 2\pi$ such that for each $z \in \mathbb C$ there exists at least one pair $(\phi,\theta)$ such that $z = \phi e^{i\theta}$. To $z = 0$ one may associate any pair of the form $(0,\theta)$ and to $z \ne 0$ one always has the two equivalent associated pairs $(\phi,\theta)$ and $(-\phi,\theta +\pi)$. These coordinates are useful for representing motion of point particles in a plane as influenced by central forces for a linear trajectory may yet be differentiable. There are no problems with the algebra represented by such coordinates so long as one is consistent about representation and it will be seen that we need not consider any possible subtleties or issues.

	Consider the stationary nonlinear Schr\"odinger equation on a line
	\begin{align}\label{cubic}
	\mu\psi = -\partial^2_x\psi -(\sigma +1)\nu|\psi|^{2\sigma}\psi,\quad \mu\in \mathbb R,\ 0 \ne \nu \in \mathbb R,\ \sigma \in \mathbb N,\ \psi \in C^2(\mathbb R, \mathbb C).
	\end{align}
	By using extended polar coordinates one may find
	\begin{align}
	&0 = \partial^2_x\psi +\mu\psi +(\sigma +1)\nu|\psi|^{2\sigma}\psi = e^{i\theta}[\partial^2_x\phi +2i\partial_x\phi\partial_x\theta -\phi(\partial_x\theta)^2 +i\phi\partial^2_x\theta]\\
	&\qquad\qquad +e^{i\theta}\mu\phi +2e^{i\theta}\nu\phi^{2\sigma},\\
	&0 = [\partial^2_x\phi +2i\partial_x\phi\partial_x\theta -\phi(\partial_x\theta)^2 +i\phi\partial^2_x\theta] +\mu\phi +(\sigma +1)\nu\phi^{2\sigma +1}.
	\end{align}
	By taking the imaginary and real parts of this equation one arrives at equations which are respectively equivalent to the angular and radial equations of motion of a Newtonian point particle moving in a planar, anharmonic, central force
	\begin{align}
	&\text{(Im.):}\quad 0 = 2\partial_x\phi\partial_x\theta +\phi\partial^2_x\theta,\\
	&\text{(Re.):}\quad 0 = \partial^2_x\phi -\phi(\partial_x\theta)^2 +\mu\phi +(\sigma +1)\nu\phi^{2\sigma +1}.
	\end{align}
	Integrating these respectively gives analogues of conservation of angular momentum and energy:
	\begin{align}
	&\text{(Im.)}:\quad \mathrm{const.} = \omega = \phi^2\partial_x\theta,\quad \omega\in\mathbb R \\
	&\text{(Re.)}:\quad \mathrm{const.} = h = (\partial_x\phi)^2 +\omega\phi^{-2} +\mu\phi^2 +\nu\phi^{2(\sigma +1)},\quad h\in\mathbb R.
	\end{align}

	The system is equivalent to that of a particle moving in the plane, with the exception that the particle might transition from one plane to the adjoined one, i.e. $\phi \mapsto -\phi$, if it passes through the origin. If $\omega \ne 0$ then the centrifugal potential energy becomes arbitrarily large as the particle moves closer to $\phi = 0$. If $|h| < \infty$ the centrifugal potential energy cannot be overcome. Solutions with $\omega = 0$ are different from those with $\omega \ne 0$ in that the former are not differentiable in standard polar coordinates, hence our introduction of the extended polar coordinates. These observations can be summarized as follows.
	
	\begin{rem}
	If $\psi = (\phi, \theta)$ is a solution of \eqref{cubic} on a line, then $\psi(x)$ can vanish for some $x$ only if $[\phi(x)]^2\partial_x\theta(x) = \omega = 0$ for all $x \in \mathbb R$.
	\end{rem}
	
	We henceforth take $\omega = 0$ and $\theta = 0$ everywhere and consider only the real solutions and then \eqref{cubic} becomes \eqref{line} on a line. This restricts our focus to all solutions that can attain zeros, and possibly a few more, at the cost of a wide class of solutions that feature nontrivial complex oscillation without attaining zeros. With respect to the effective particle total energy, one then has
	\begin{align}\label{energy}
	h &= (\partial_x\phi)^2 +\mu\phi^2 +\nu\phi^{2(\sigma +1)},\quad h,\mu \in \mathbb R,\ 0 \ne \nu \in \mathbb R,\ \sigma \in \mathbb N,\ \phi \in C^2(\mathbb R, \mathbb R).
	\end{align}
	One may partition the effective particle total energy $h$ respectively into kinetic and potential parts, $h_\mathrm k$ and $h_\mathrm p$, via
	\begin{align}
	h &= h_\mathrm k +h_\mathrm p,\quad h_\mathrm k(\phi) := (\partial_x\phi)^2,\quad h_\mathrm p(\phi) := \mu\phi^2 +\nu\phi^{2(\sigma +1)}.
	\end{align}
	
	We have now reduced the system to that of a classical point particle constrained moving along a potential energy surface with constant total energy. This allows us to classify all solutions of \eqref{line} on a line. Fix $\nu \ne 0$, we denote
	\begin{align}
	\Phi_\mathrm{line} &:= \{\phi \ne 0\text{ solves \eqref{line} on a line, such that }\phi\text{ is periodic and attains zeros}\}.
	\end{align}
	
	\begin{prop}\label{sol_space}
	Let $\nu \ne 0$. Given $\phi \in \Phi_\mathrm{line}$, there is a unique value of $\mu \in \mathbb R$ such that $\phi$ is a solution of \eqref{line} on the line with the given values $\mu, \nu$. This allows us to define the map
	\begin{align}
	\Lambda_\mathrm{line}: \Phi_\mathrm{line} \to P,\quad \Lambda_\mathrm{line}: \phi \mapsto (\mu, ||\phi||_\infty),
	\end{align}
	which is onto and $\Lambda_\mathrm{line}(\phi_1) = \Lambda_\mathrm{line}(\phi_2)$ if and only if $\phi_1(x) = \zeta\phi_2(x +x_0)$ for some fixed $\zeta = \pm1$, some fixed $x_0 \in \mathbb R$, and all $x$.
	\end{prop}

	\begin{proof}
	
	\underline{Solution theory via energy conservation}.
	
	First we prove that $\Lambda_\mathrm{line}$ is onto and show the degree of freedom in its preimages. At the end of the proof we show the uniqueness of $\mu$. Solutions of \eqref{line} follow from conservation of effective particle total energy and qualitative analysis of dynamics through determination of critical points of the effective particle motion as follows. From \eqref{energy} we get
	\begin{align}
	\widehat x(\widehat\phi_0, \widehat\phi) &= x_0 +\zeta\int_{\widehat\phi_0}^{\widehat\phi}\mathrm dw\ [h -\mu w^2 -\nu w^{2(\sigma +1)}]^{-1/2},\label{integral}
	\end{align}
	where $\zeta = \pm1$, $x_0 \in \mathbb R$ is an initial value of $x$ along a trajectory and
	\begin{align}
	\widehat x(\widehat\phi_0, \widehat\phi_0) = x_0,\quad \widehat x(\widehat\phi_0, \widehat\phi) = x,\quad \phi(x_0) = \widehat\phi_0,\quad \phi(x) = \widehat\phi.
	\end{align}
	The map $\widehat x(\widehat\phi_0, \widehat\phi)$ presents an inverse function for the solution, $\phi: x \mapsto \phi(x)$, that is defined piecewise between the obstructions $\partial_x\phi = 0$. Since $\widehat x(\widehat\phi_0,\widehat\phi)$ is necessarily monotone in $\widehat\phi$ between these obstructions, the function may be inverted on these intervals to recover $\phi(x)$. The solutions can be continued past the obstructions by adjoining the piecewise solutions in a manner that satisfies \eqref{line} and energy conservation appropriately.
	
	The turning values are specified by the values of $\phi$, which we denote by $\phi(x) = \beta$, and satisfy
	\begin{align}
	h &= \mu\beta^2 +\nu\beta^{2(\sigma +1)}.
	\end{align}
	These are illustrated in Figures \ref{case_1_h}-\ref{case_3_4}. For $\sigma = 1$, one may find
	\begin{align}
	\beta_n^2 &= -2^{-1}\nu^{-1}[\mu +(-1)^n(\mu^2 +4h\nu)^{1/2}] ,\quad n = 1,2.
	\end{align}
	For other values of $\sigma$, calculation of the $\beta_n$ might not be so straightforward but they can be assured to exist due to the simple local monotonicity properties of $h_\mathrm p(\cdot)$ and thereby are also qualitatively similar to the values for $\sigma = 1$ in that they appear in pairs that are real, imaginary, or otherwise accordingly.
	
	To show that $\Lambda_\mathrm{line}$ is onto and study its preimages, it is helpful to first exhaustively classify the solutions of \eqref{line} on a line up to translation $x \mapsto x +x_0$, which may also be seen in \cite{GnWa16}, and relate the results to the auxiliary parameter $\alpha$ where possible. We do so by considering the distinguished parameter regions for the effective particle potential energy while recalling that the effective particle kinetic energy $h_\mathrm k(\phi)$ is necessarily nonnegative.
	
	\begin{figure}[t]
	\centering
	\includegraphics[width=.5\textwidth]{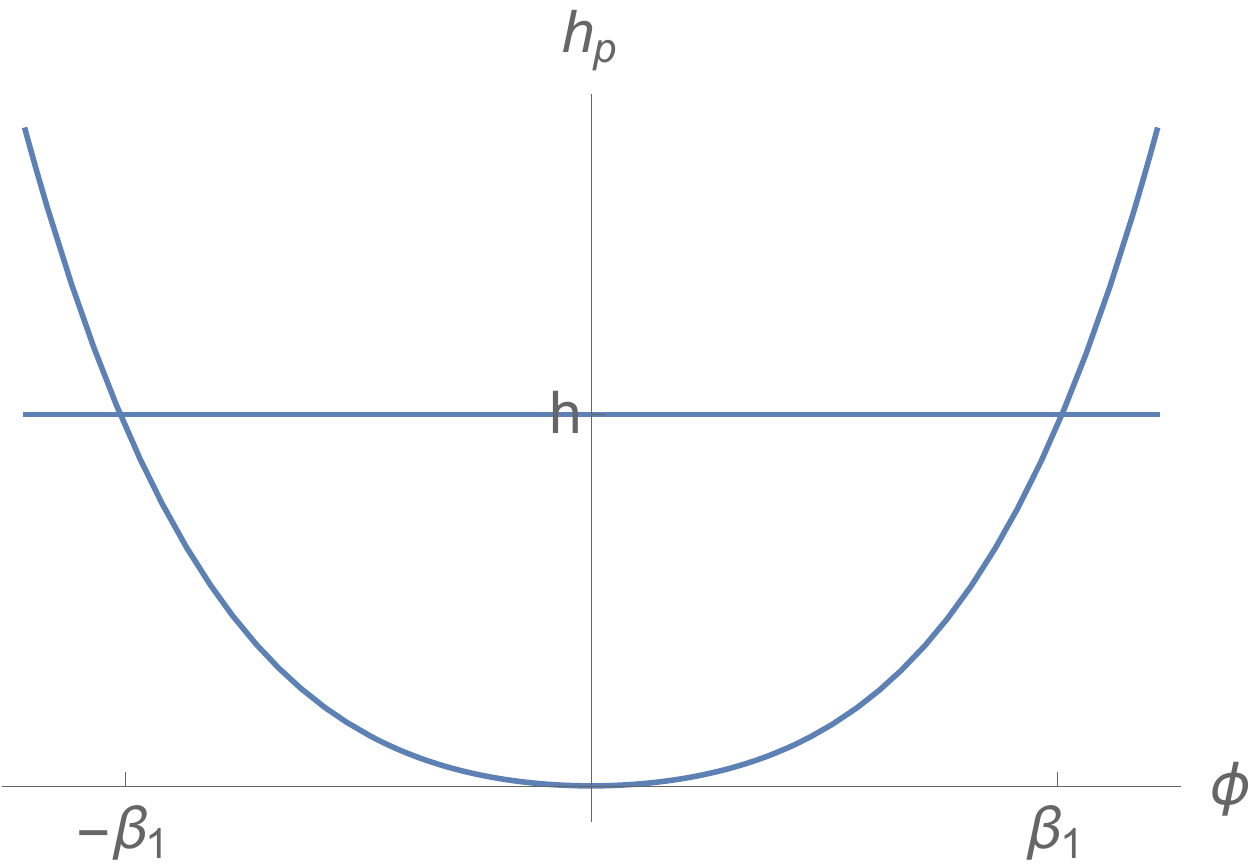}
	\caption{Case I, $(\nu > 0, \mu \ge 0)$: $h > 0$}\label{case_1_h}
	\end{figure}

	\underline{Case I, $(\nu > 0, \mu \ge 0)$}. $h_\mathrm p(\phi)$ has a global minimum at $\phi = 0$ and $\lim_{\phi \to \pm\infty}h_\mathrm p(\phi) = \infty$. There are three notable ranges of $h$.
	\begin{enumerate}
	\item $h < 0$. There are no solutions.
	\item $h = 0$. There is only the constant solution $\phi(x) = 0$.
	\item $h > 0$. Shown in Figure \ref{case_1_h}. There is only the solution which oscillates as $-\beta_1 \le \phi(x) \le \beta_1$. This solution is bounded, periodic, attains zeros and satisfies $\alpha = \beta_1 > 0$.
	\end{enumerate}
	Then for Case I, elements of $\Phi_\mathrm{line}$ may belong only to the sub-case (3) above, for which $h > 0$, and $(\mu, \alpha) \in P^+_+$ for these solutions.
	
	\begin{figure}[t]
	\centering
	\subfigure[$h_\mathrm c < h < 0$]{\includegraphics[width=.47\textwidth]{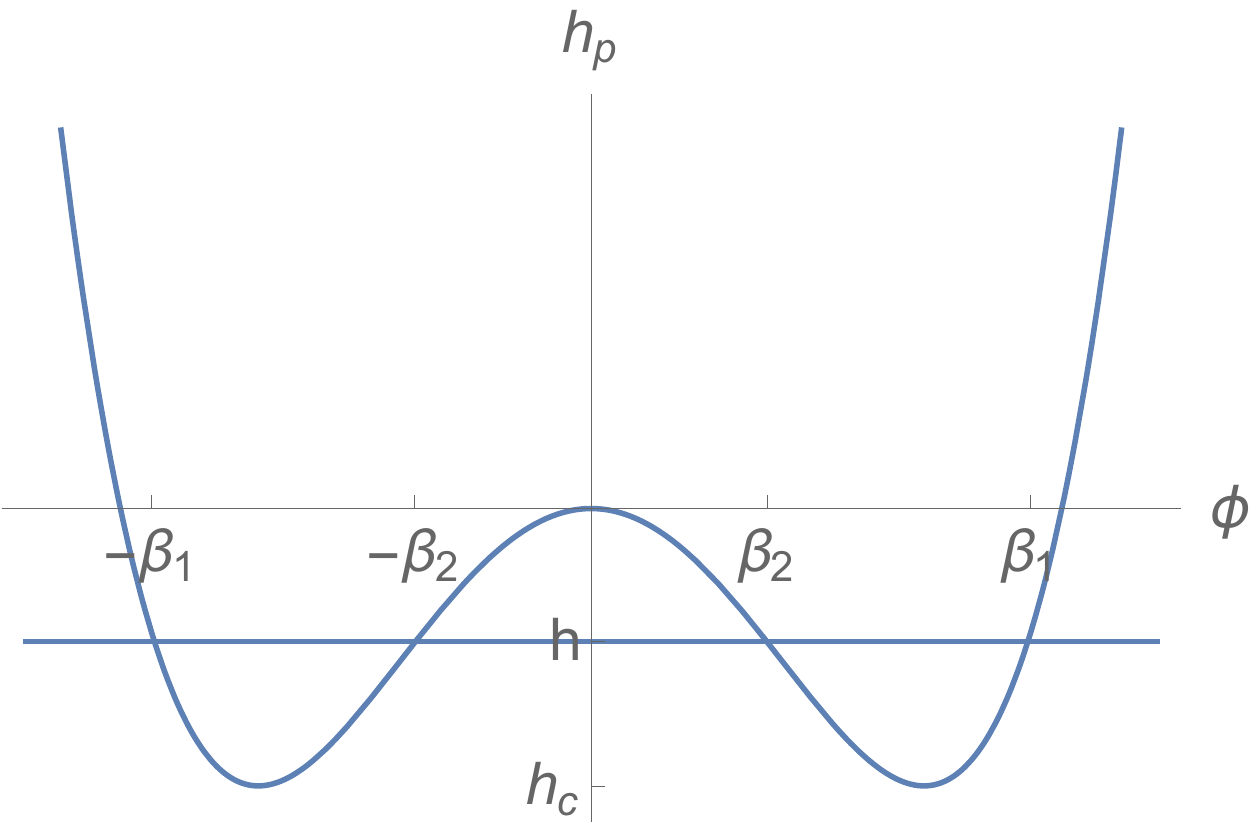}}
	\hfill
	\subfigure[$h > 0$]{\includegraphics[width=.47\textwidth]{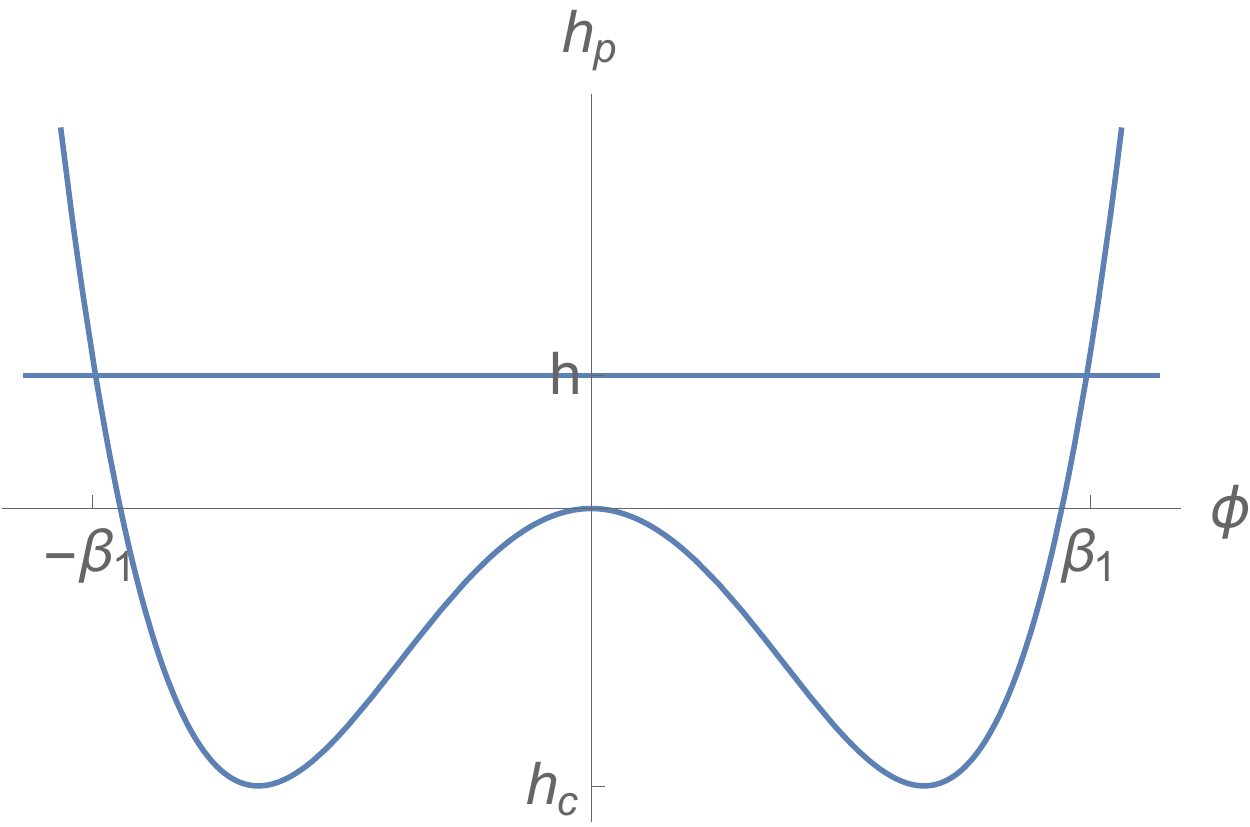}}
	\caption{Case II, $(\nu > 0, \mu < 0)$}\label{case_2_h}
	\end{figure}
	
	\underline{Case II, $(\nu > 0, \mu < 0)$}. $h_\mathrm p(\phi)$ has a local maximum at $\phi = 0$, global minima at $\phi = \pm\alpha_{\mathrm c}$, and $\lim_{\phi \to \pm\infty}h_\mathrm p(\phi) = \infty$. There are five notable ranges of $h$.
	\begin{enumerate}
	\item $h < h_\mathrm c$. There are no solutions.
	\item $h = h_\mathrm c$. There are only the two constant solutions $\phi = \pm\alpha_\mathrm c$.
	\item $h_\mathrm c < h < 0$. Shown in Figure \ref{case_2_h} (a). There are two solutions. Each has definite sign and are negatives of one another. The positive solution oscillates as $\beta_2 \le \phi(x) \le \beta_1$. These solutions are bounded, periodic, and attain no zeros.
	\item $h = 0$. There are two solutions. They are ``soliton solutions'' and are negatives of one another. They have the maximum absolute value $||\phi||_\infty = \alpha_0$. One is strictly positive and for $\phi(0) = \alpha_0$ it satisfies $\phi(x) \to 0$ monotonically as $0 < x \to \infty$ and $\phi(-x) = -\phi(x)$ since by \eqref{integral}
	\begin{align}
	&\left[\lim_{\epsilon \searrow 0} \widehat x(\epsilon, \widehat\phi) -x_0\right]/\zeta = \lim_{\epsilon \searrow 0} \int_\epsilon^{\widehat\phi}\mathrm dw\ (|\mu|w^2 -|\nu|w^{2(\sigma +1)})^{-1/2}\label{infty_1}\\
	&\qquad = \lim_{\epsilon \searrow 0} \int_\epsilon^{\widehat\phi}\mathrm dw\ w^{-1}(|\mu| -|\nu|w^{2\sigma})^{-1/2} = \infty.\label{infty_2}
	\end{align}
	These solutions are bounded, not periodic, and attain no zeros.
	\item $h > 0$. Shown in Figure \ref{case_2_h} (b). There is only the solution which oscillates as $-\beta_1 \le \phi(x) \le \beta_1$. This solution is bounded, periodic, attains zeros and satisfies $\alpha = \beta_1 > \alpha_0$.
	\end{enumerate}
	Then for Case II, elements of $\Phi_\mathrm{line}$ may belong only to the sub-case (5) above, for which $h > 0$, and $(\mu, \alpha) \in P^+_-$ for these solutions.
	
	\begin{figure}[t]
	\centering
	\subfigure[Case III, $(\nu < 0, \mu \ge 0)$: $0 < h < h_\mathrm c$]{\includegraphics[width=.47\textwidth]{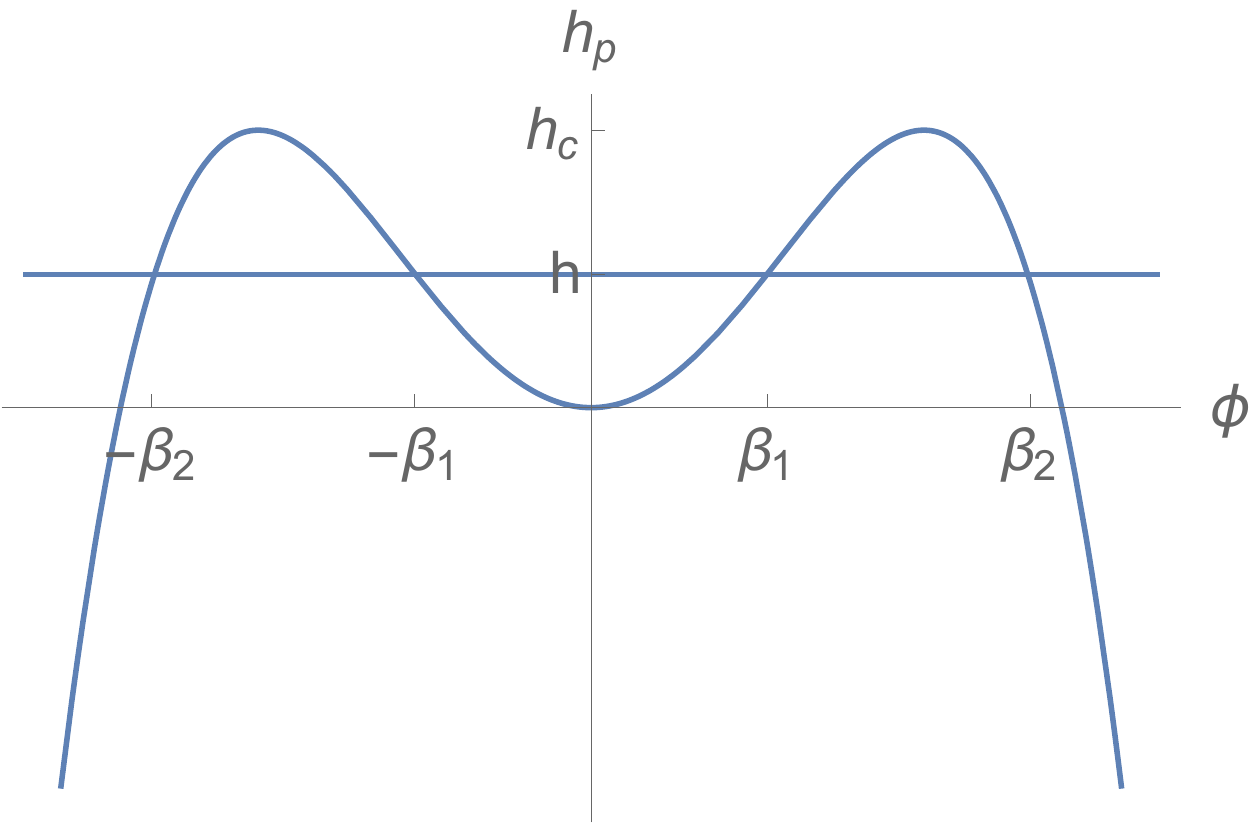}}
	\hfill
	\subfigure[Case IV, $(\nu < 0, \mu < 0)$: $h < 0$]{\includegraphics[width=.47\textwidth]{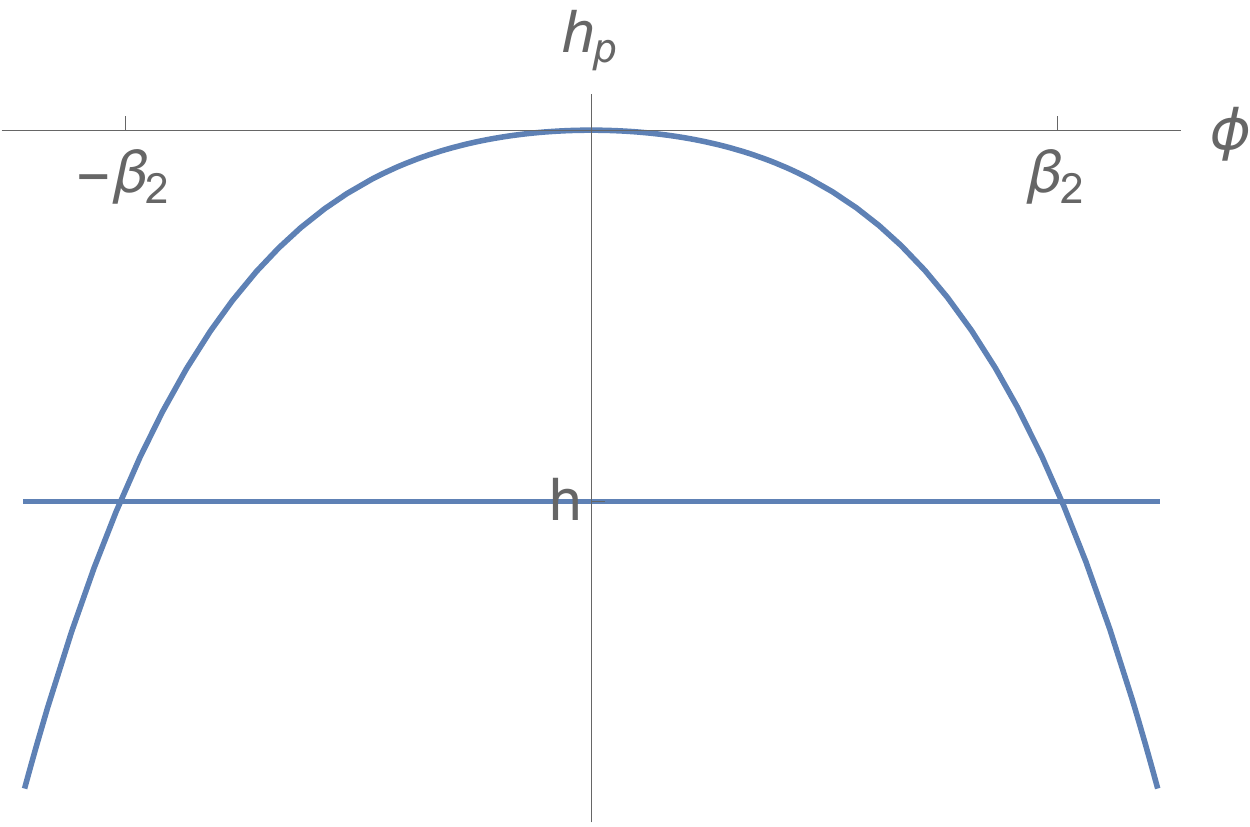}}
	\caption{Some subcases of Cases III and IV}\label{case_3_4}
	\end{figure}
	
	\underline{Case III, $(\nu < 0, \mu \ge 0)$}. $h_\mathrm p(\phi)$ has a local minimum at $\phi = 0$, global maxima at $\phi = \pm\alpha_{\mathrm c}$, and $\lim_{\phi \to \pm\infty}h_\mathrm p(\phi) = -\infty$. There are five notable ranges of $h$.
	\begin{enumerate}
	\item $h < 0$. There are two solutions. Each have definite sign and are negatives of one another. The positive solution has a minimum value $\phi = \beta_2$. For $\phi(0) = \beta_2$, one has that $\phi(x) \to \infty$ monotonically as $0 \le x \to \infty$ and $\phi(-x) = -\phi(x)$.
	\item $h = 0$. There are three solutions. Two are analogous to those for $h < 0$ and the remaining one is the constant solution $\phi = 0$.
	\item $0 < h < h_\mathrm c$. Shown in Figure \ref{case_3_4} (a). There are three solutions. Two are analogous to those for $h < 0$. The remaining one oscillates as $-\beta_1 \le \phi(x) \le \beta_1$. The oscillating solution is bounded, periodic, attains zeros, and satisfies $\alpha = \beta_1 < \alpha_\mathrm c$.
	\item $h = h_\mathrm c$. There is are two solutions. They are ``kink solutions'' and are negatives of one another. One is strictly increasing in $x$, $\phi(x) \to \alpha_\mathrm c$ monotonically as $x \to \infty$, and satisfies $\phi(-x) = -\phi(x)$ for $\phi(0) = 0$, which can be shown by a calculation similar to that of \eqref{infty_1} and \eqref{infty_2}. These solutions are bounded, not periodic, and attain only one zero.
	\item $h_\mathrm c < h$. There are two solutions. They are negatives of one another. One is strictly increasing in $x$, $\phi(x) \to \infty$ monotonically as $x \to \infty$, and satisfies $\phi(-x) = -\phi(x)$ for $\phi(0) = 0$.
	\end{enumerate}
	Then for Case III, elements of $\Phi_\mathrm{line}$ may belong only to the sub-case (3) above, for which $0 < h < h_{\mathrm c}$, and $(\mu, \alpha) \in P_-$ for these solutions.
	
	\underline{Case IV, $(\nu < 0, \mu < 0)$}. $h_\mathrm p(\phi)$ has a global maximum at $\phi = 0$ and $\lim_{\phi \to \pm\infty}h_\mathrm p(\phi) = -\infty$. There are three notable ranges of $h$.
	\begin{enumerate}
	\item $h < 0$. Shown in Figure \ref{case_3_4} (b). There are two solutions. Each have definite sign and are negatives of one another. The positive solution has a minimum value $\phi(x) = \beta_2$. Without loss of generality take $\phi(0) = \beta_2$. One has that $\phi(x) \to \infty$ monotonically as $0 \le x \to \infty$ and $\phi(-x) = -\phi(x)$.
	\item $h = 0$. There is only the constant solution $\phi = 0$.
	\item $0 < h$. There are two solutions. They are negatives of one another. One is strictly increasing in $x$, $\phi(x) \to \infty$ monotonically as $x \to \infty$, and satisfies $\phi(-x) = -\phi(x)$ for $\phi(0) = 0$.
	\end{enumerate}
	Then consideration of Case IV shows that none of its solutions belong to $\Phi_\mathrm{line}$.
	
	\underline{Conclusion}
	
	By the exhaustive classification of solutions the map $\Lambda_\mathrm{line}$ must be onto, and $\Lambda_\mathrm{line}(\phi_1) = \Lambda_\mathrm{line}(\phi_2)$ if and only if $\phi_1(x) = \zeta\phi_2(x +x_0)$ for some fixed $\zeta = \pm1$, some fixed $x_0 \in \mathbb R$, and all $x$.
	
	\underline{Different solutions if and only if different values of $\mu$}.
	
	Assume that $\phi \in \Phi$ and take $x \in \mathbb R$ to satisfy at least one of $\phi(x), \partial_x\phi(x)$ differs from zero. Such an $x$ exists by the classification done above. By observing the RHS of \eqref{line} one can see that a given $\phi$ uniquely specifies the $\mu$ with which it is associated. This proves that $\Lambda_\mathrm{line}$ is well defined.
	
	\end{proof}
	
	The classification made above for solutions in $\Phi_\mathrm{line}$ allows us to study the map $\Lambda_\mathrm{int}$. We may now prove Lemma \ref{lem_int}.
	
	\begin{proof}[Proof of Lemma \ref{lem_int}]
	\underline{Part \eqref{lem_int_1}.}
	Proving that $\Lambda_\mathrm{int}$ is well defined follows by the same argument as that which was used for $\Lambda_\mathrm{line}$, as was implemented above in the proof of Proposition \ref{sol_space}.
	
	\underline{Part \eqref{lem_int_2}.} Let $\phi_1, \phi_2 \in \Phi_\mathrm{int}$ such that $\Lambda_\mathrm{int}(\phi_1) = \Lambda_\mathrm{int}(\phi_2)$. By Part \eqref{lem_int_1} of the Lemma, $\phi_1, \phi_2$ share the same values of $\mu, \alpha$. Hence both of them correspond to trajectories of a classical particle moving in the same potential, which belongs to one of the four cases in proof of Proposition \ref{sol_space}. Having the same value of $\alpha$ means that both trajectories have the same energy, $h$ as seen in \eqref{energy}. Pick one boundary point of the interval. If the boundary condition at this point is Dirichlet, then both trajectories start at $\phi(x) = 0$ and since they have equal energies, their initial velocities are the same up to a sign, from which we conclude that the trajectories are equal up to a sign, i.e. $\phi_1 = \zeta\phi_2$ where $\zeta = \pm1$. Alternatively, if the boundary condition is Neumann, then both trajectories start at $\phi(x) = \pm\alpha$ with zero velocity and once again it implies that they are equal up sign, i.e. $\phi_1 = \zeta\phi_2$ where $\zeta = \pm1$. The proof is finished once we note that $\Lambda_\mathrm{int}(\phi) = \Lambda_\mathrm{int}(-\phi)$.
	\end{proof}

	Next is an easy but important Lemma that establishes a connection between solutions of \eqref{line} on an interval and on a line.	
	
	\begin{lem}\label{restrict_to_int}
	Every solution $\phi \in \Phi_\mathrm{int}$ is a restriction of a solution $\widehat\phi = \Lambda_\mathrm{line}^{-1}\circ\Lambda_\mathrm{int}(\phi)$ from the line to an appropriate interval, where $||\widehat\phi||_\infty = ||\phi||_\infty$. Note that $\Lambda_\mathrm{line}$ is not surjective and hence $\Lambda_\mathrm{line}^{-1}$ is not uniquely defined however, for the sake of the statement, any preimage of $\Lambda_\mathrm{line}$ can be chosen as the image of $\Lambda_\mathrm{line}^{-1}$.
	\end{lem}
	
	\begin{proof}
	Given a solution $\phi \in \Phi_\mathrm{int}$ we apply $\Lambda_\mathrm{int}$ to get the corresponding $(\mu, \alpha)$. By the classification of solutions done in the proof of Proposition \ref{sol_space}, this $(\mu, \alpha)$ corresponds to a particle trajectory on the line, $\widehat\phi \in \Phi_\mathrm{line}$. Our solution $\phi$ serves as a subtrajectory and hence can be obtained as a restriction of $\widehat\phi$. The equality $||\widehat\phi||_\infty = ||\phi||_\infty$ follows as the trajectory $\phi$ always attains the maximal absolute value of the trajectory $\widehat\phi$ either at an endpoint if there is a Neumann condition there or somewhere in between if both boundary conditions are Dirichlet.
	\end{proof}

	\section{The wavelength $\lambda$}\label{wavelength}
	
	The elements of $\Phi_\mathrm{line}$ are periodic. We call this period the wavelength and denote it by $\lambda$. In this section we study the dependence of $\lambda$ in the parameters $\mu, \alpha$, which would allow the classification of solutions in $\Phi_\mathrm{int}$.
	
	\begin{defn}
	Fix $\nu \ne 0$, for $(\mu, \alpha) \in P$ denote:
	\begin{align}
	\kappa_{(\mu, \alpha)}(w) := [\mu(1 -w^2) +\nu\alpha^{2\sigma}(1 -w^{2(\sigma +1)})]^{-1/2},\quad w \in [0, 1].
	\end{align}
	\end{defn}
	
	\begin{prop}\label{properties}
	For each $(\mu, \alpha) \in P$, the solution $\phi = \phi_{(\mu, \alpha)} \in \Phi_\mathrm{line}$ is periodic with wavelength (period) $\lambda \equiv \lambda_{(\mu, \alpha)}$ of the form
	\begin{align}
	\lambda_{(\mu, \alpha)} = 4\int_0^1\mathrm dw\ \kappa_{(\mu, \alpha)}(w)\label{lambda}
	\end{align}
	and that satisfies the following properties.
	
	\underline{$(\mu, \alpha) \in P^+_+$}:
	\begin{align}
	&\lim_{\mu \to \infty} \lambda_{(\mu, \alpha)} = 0,\quad \lim_{\alpha \to 0} \lambda_{(\mu, \alpha)} = 2\pi\mu^{-1/2},\quad \partial_\mu\lambda_{(\mu, \alpha)} < 0,\quad \partial_\alpha\lambda_{(\mu, \alpha)} < 0,\label{lambda_a}
	\end{align}
	
	\underline{$(\mu, \alpha) \in P^+_-$}:
	\begin{align}
	&\lim_{\mu \searrow \mu_0}\lambda_{(\mu, \alpha)} = \infty,\quad \partial_\mu\lambda_{(\mu, \alpha)} < 0,\quad \partial_\alpha\lambda_{(\mu, \alpha)} < 0,\label{lambda_b}
	\end{align}
	
	\underline{$(\mu, \alpha) \in P^-$}:
	\begin{align}
	&\lim_{\mu \to \infty} \lambda_{(\mu, \alpha)} = 0,\quad \lim_{\alpha \to 0}\lambda_{(\mu, \alpha)} = 2\pi\mu^{-1/2},\quad \lim_{\mu \searrow \mu_\mathrm c} \lambda_{(\mu, \alpha)} = \infty,\label{lambda_c}\\
	&\qquad \partial_\mu\lambda_{(\mu, \alpha)} < 0,\quad \partial_\alpha\lambda_{(\mu, \alpha)} > 0,\label{lambda_d}
	\end{align}
	where $\mu_0 = -|\nu|\alpha^{2\sigma}$, $\mu_{\mathrm c} = (\sigma +1)|\nu|\alpha^{2\sigma}$.
	
	\end{prop}

	\begin{proof}
	\underline{Proof of representation of $\lambda$ in \eqref{lambda}}. For solutions of the form $\phi = \phi_{(\mu, \alpha)} \in \Phi_\mathrm{line}$ for, $(\mu, \alpha) \in P$, one may calculate the quarter wavelength through
	\begin{align}
	&\lambda_{(\mu, \alpha)}/4 = \int_0^\alpha\mathrm dw\ [h -\mu w^2 -\nu w^{2(\sigma +1)}]^{-1/2}\\
	&\qquad = \int_0^1\mathrm dw\ \alpha[h -\mu\alpha^2w^2 -\nu\alpha^{2(\sigma +1)}w^{2(\sigma +1)}]^{-1/2}\\
	&\qquad = \int_0^1\mathrm dw\ [\alpha^{-2}h -\mu w^2 -\nu\alpha^{2\sigma}w^{2(\sigma +1)}]^{-1/2}.
	\end{align}
	By following the analogy to particle dynamics, one may consider $\lambda$ to be the particle period of oscillation in time.

	Since the effective particle total energy satisfies
	\begin{align}
	h &= (\partial_x\phi)^2 +\mu\phi^2 +\nu\phi^{2(\sigma +1)}
	\end{align}
	one has
	\begin{align}
	h &= \mu \alpha^2 +\nu \alpha^{2(\sigma +1)}
	\end{align}
	and therefore
	\begin{align}
	\lambda_{(\mu, \alpha)}/4 &= \int_0^1\mathrm dw\ [\mu(1 -w^2) +\nu\alpha^{2\sigma}(1 -w^{2(\sigma +1)})]^{-1/2} = \int_0^1\mathrm dw\ \kappa_{(\mu, \alpha)}(w),
	\end{align}
	which proves the representation of $\lambda$ in \eqref{lambda}.

	\underline{Proof of regular limits in \eqref{lambda_a}, \eqref{lambda_b}, \eqref{lambda_c}}. Consider that $\nu > 0$. By inspection we have that $\lim_{\mu \to \infty} \kappa_{(\mu, \alpha)} = 0$ uniformly, and therefore $\lim_{\mu \to \infty} \lambda_{(\mu, \alpha)} = 0$, which proves the first relations of \eqref{lambda_a} and \eqref{lambda_c}.
	
	One may write
	\begin{align}
	\kappa_{(\mu, \alpha)}(w) &= \left(\mu +\nu\alpha^{2\sigma}\sum_{j = 0}^\sigma w^{2j}\right)^{-1/2}(1 -w^2)^{-1/2},
	\end{align}
	which is a form that is very useful for analysis of the integral. For $\nu > 0$ one has
	\begin{align}
	\kappa_{(\mu, \alpha)}(w) &= \left(\mu +\nu\alpha^{2\sigma}\sum_{j = 0}^\sigma w^{2j}\right)^{-1/2}(1 -w^2)^{-1/2} \le (\mu +\nu\alpha^{2\sigma})^{-1/2}(1 -w^2)^{-1/2}
	\end{align}
	for all $w \in (0, 1)$ and for $\nu < 0$ one has
	\begin{align}
	\kappa_{(\mu, \alpha)}(w) &\le [\mu +(\sigma +1)\nu\alpha^{2\sigma}]^{-1/2}(1 -w^2)^{-1/2}
	\end{align}
	for all $w \in (0, 1)$, both bounds of $\kappa_{(\mu, \alpha)}(w)$ are positive and integrable. These bounds justify the application of dominated convergence theorem to evaluate
	\begin{align}
	&\mu > 0 :\ \lim_{\alpha \to 0}\lambda_{(\mu, \alpha)} = \int_0^1\lim_{\alpha \to 0}\kappa_{(\mu, \alpha)} = 4\mu^{-1/2}\int_0^1(1 -w^2)^{-1/2} = 2\pi\mu^{-1/2},
	\end{align}
	which proves the second relations in \eqref{lambda_a} and \eqref{lambda_c}.
	
	\underline{Proof of singular limits in \eqref{lambda_b}, \eqref{lambda_c}}. Take $\nu > 0$ and $\mu < 0$. We recall that $\mu_0 = -|\nu|\alpha^{2\sigma}$ and observe that for $\mu = \mu_0 +\epsilon$, $\epsilon > 0$, one has
	\begin{align}
	&\kappa_{(\mu, \alpha)}(w) = [\mu(1 -w^2) +\nu\alpha^{2\sigma}(1 -w^{2(\sigma +1)})]^{-1/2}\\
	&\qquad = (\mu -\mu w^2 +\nu\alpha^{2\sigma} -\nu\alpha^{2\sigma}w^{2(\sigma +1)})^{-1/2}\\
	&\qquad = (\epsilon -\mu_0 w^2 -\epsilon w^2 -\nu\alpha^{2\sigma}w^{2(\sigma +1)})^{-1/2}\\
	&\qquad \ge (\epsilon +|\mu_0|w^2)^{-1/2} = |\mu_0|^{-1/2}(\epsilon/|\mu_0| +w^2)^{-1/2}
	\end{align}
	and then by direct calculation
	\begin{align}
	&\lim_{\epsilon \to 0}\lambda_{(\mu, \alpha)}/4 \ge \lim_{\epsilon \to 0}\int_0^1\mathrm dw\ |\mu_0|^{-1/2}(\epsilon/|\mu_0| +w^2)^{-1/2} = \infty,
	\end{align}
	which proves the first relation in \eqref{lambda_b}.
	
	Take $\nu < 0$ and $\mu > 0$. We recall that $\mu_{\mathrm c} = (\sigma +1)|\nu|\alpha^{2\sigma}$ and observe that for $\mu = \mu_{\mathrm c}(1 -\epsilon)^{-2\sigma}$, $\epsilon > 0$, one has
	\begin{align}
	&\kappa_{(\mu, \alpha)}(w) = \left(\mu +\nu\alpha^{2\sigma}\sum_{j = 0}^\sigma w^{2j}\right)^{-1/2}(1 -w^2)^{-1/2}\\
	&\qquad = \left[\mu_{\mathrm c}(1 -\epsilon)^{-2\sigma} -|\nu|\alpha^{2\sigma}\sum_{j = 0}^\sigma w^{2j}\right]^{-1/2}(1 -w^2)^{-1/2}\\
	&\qquad \ge \left[\mu_{\mathrm c}(1 -\epsilon)^{-2\sigma} -(\sigma +1)|\nu|\alpha^{2\sigma}w^{2\sigma}\right]^{-1/2}(1 -w^2)^{-1/2}\\
	&\qquad = [(1 -\epsilon)^{2\sigma}/\mu_{\mathrm c}]^{1/2}[1 -(1 -\epsilon)^{2\sigma}w^{2\sigma}]^{-1/2}(1 -w^2)^{-1/2}
	\end{align}
	\begin{align}
	&\qquad \ge \mu_{\mathrm c}^{-1/2}(1 -\epsilon)^\sigma[1 -(1 -\epsilon)^{2\sigma}w^{2\sigma}]^{-1/2}[1 -(1 -\epsilon)^{2\sigma}w^{2\sigma}]^{-1/2}\\
	&\qquad = \mu_{\mathrm c}^{-1/2}(1 -\epsilon)^\sigma[1 -(1 -\epsilon)^{2\sigma}w^{2\sigma}]^{-1}\\
	&\qquad = \mu_{\mathrm c}^{-1/2}(1 -\epsilon)^\sigma\left\{\sum_{j = 0}^{2\sigma -1}[(1 -\epsilon)w]^j\right\}^{-1}[1 -(1 -\epsilon)w]^{-1}\\
	&\qquad \ge (2\sigma)^{-1}\mu_{\mathrm c}^{-1/2}(1 -\epsilon)^\sigma[1 -(1 -\epsilon)w]^{-1},
	\end{align}
	
	and then by direct calculation
	\begin{align}
	&\lim_{\epsilon \to 0}\lambda_{(\mu, \alpha)}/4 \ge \lim_{\epsilon \to 0}\int_0^1\mathrm dw\ (2\sigma)^{-1}\mu_{\mathrm c}^{-1/2}(1 -\epsilon)^\sigma[1 -(1 -\epsilon)w]^{-1}\\
	&\qquad = (2\sigma)^{-1}\mu_{\mathrm c}^{-1/2}\lim_{\epsilon \to 0}\int_0^1\mathrm dw\ [1 -(1 -\epsilon)w]^{-1}\\
	&\qquad = (2\sigma)^{-1}\mu_{\mathrm c}^{-1/2}\lim_{\epsilon \to 0}[-(1 -\epsilon)^{-1}\log(1 -\{1 -\epsilon\}w)]_0^1 = \infty,
	\end{align}
	which proves the third relation in \eqref{lambda_c}.
	
	\underline{Proof of monotonicity in \eqref{lambda_a}, \eqref{lambda_b}, \eqref{lambda_d}}.

	By inspection one can see that $\kappa_{(\mu, \alpha)}(w)$, $\partial_\mu\kappa_{(\mu, \alpha)}(w)$, $\partial_\alpha\kappa_{(\mu, \alpha)}(w)$ are continuous in $\mu, \alpha, w$, and therefore by Leibniz integral rule:
	\begin{align}
	\partial_\mu \int_0^1 \mathrm dw\ \kappa_{(\mu, \alpha)}(w) &= \int_0^1 \mathrm dw\ \partial_\mu \kappa_{(\mu, \alpha)}(w),\quad \partial_\alpha \int_0^1 \mathrm dw\ \kappa_{(\mu, \alpha)}(w) = \int_0^1 \mathrm dw\ \partial_\alpha \kappa_{(\mu, \alpha)}(w).
	\end{align}
	One may thereby determine the signs of $\partial_\mu\lambda_{(\mu, \alpha)}$ and $\partial_\alpha\lambda_{(\mu, \alpha)}$ from
	\begin{align}
	\partial_\mu\kappa_{(\mu, \alpha)}(w) &= -2^{-1}\left(\mu +\nu\alpha^{2\sigma}\sum_{j = 0}^\sigma w^{2j}\right)^{-3/2}(1 -w^2)^{-1/2}\\
	\partial_\alpha\kappa_{(\mu, \alpha)}(w) &= -\sigma\nu\alpha^{2\sigma -1}\sum_{j = 0}^\sigma w^{2j}\left(\mu +\nu\alpha^{2\sigma}\sum_{j = 0}^\sigma w^{2j}\right)^{-3/2}(1 -w^2)^{-1/2},
	\end{align}
	which proves the last two relations in \eqref{lambda_a}, \eqref{lambda_b}, \eqref{lambda_d}.
	\end{proof}

	\section{Proof of Theorem \ref{thm_1}}\label{thm_1_proof}
	
	We are now prepared to directly prove our results on an interval.
	
	\begin{proof}[Proof of Theorem \ref{thm_1}]
	
	\underline{Proof of Parts \eqref{thm_1_part_1}, \eqref{thm_1_part_2}}.
	
	By Lemma \ref{restrict_to_int}, elements of $\Phi_\mathrm{int}$ may be formed by restricting elements of $\Phi_\mathrm{line}$ to functions on an interval. Such restrictions must be made so that the endpoints are zeros or local extrema as needed to satisfy the Dirichlet or Neumann boundary points. All elements of $\Phi_\mathrm{line}$ are periodic with wavelength (period) $\lambda$. Therefore $\phi = \phi_{(\mu, \alpha)} \in \Phi_\mathrm{int}$ and has a given number of $n -1$ zeros if one can ensure that $\lambda_{(\mu, \alpha)} = \lambda_n$, $n \in \mathscr N$, is chosen appropriately:
	\begin{itemize}
	\item For one matching condition of Dirichlet type and one of Neumann type one requires $\lambda_n = 4l/(2n -1)$ where $n \in \mathbb N$.
	\item For both boundary conditions of Dirichlet type one requires $\lambda_n = 2l/n$ where $n \in \mathbb N$.
	\item For both boundary conditions of Neumann type one requires $\lambda_n = 2l/n$ where $n \in \mathbb N\setminus\{1\}$.
	\end{itemize}
	We note that if both boundary conditions are Neumann, then by our construction there are no stationary solutions without zeros, hence $\mathscr N = \mathbb N\setminus\{1\}$ for this case.
	
	In each case, we must ensure that the wavelength takes on the appropriate value, $\lambda_{(\mu, \alpha)} = \lambda_n$, which guarantees that $\phi_{(\mu, \alpha)}$ has $n -1$ internal zeros. Next we show that this can be done for each $\alpha > 0$. Take to be any fixed values $\alpha > 0$ and $\nu \ne 0$. By the monotonicity properties of Proposition \ref{properties}, $f_\alpha: \mu \mapsto f_\alpha(\mu) = \lambda_{(\mu, \alpha)}$ is a monotone strictly decreasing function and, due to the limiting properties seen in Proposition \ref{properties}, that is surjective on $(0, \infty)$. Therefore there must exist only one $\mu = f_\alpha^{-1}(\lambda_n)$ for which $\lambda_{(\mu, \alpha)} = \lambda_n$ and one may write
	\begin{align}
	\gamma_n = \bigcup_{\alpha \in (0, \infty)}\{(f_\alpha^{-1}(\lambda_n), \alpha)\},
	\end{align}
	which must be a connected curve because of continuity of $f_\alpha^{-1}(\lambda_n)$ in $\alpha \in (0, \infty)$. We also observe that, due to the above representation, each $\gamma_n$ is a level curve of $\lambda_{(\mu, \alpha)}$ along $\lambda = \lambda_n$. This implies that the $\gamma_n$ are mutually disjoint. The gradient of $\lambda_{(\mu, \alpha)}$ on the $(\mu, \alpha)$ half plane is nonvanishing by the monotonicity properties of Proposition \ref{properties} and this implies that the level curves $\gamma_n$ are non self intersecting.
		
	\underline{Proof of Part \eqref{thm_1_part_3}}.
	
	From above, one may construct each spectral curve $\gamma_n$ by taking the union of points obtained from $(f_\alpha^{-1}(\lambda_n), \alpha)$ by allowing $\alpha$ to vary on $(0, \infty)$. Therefore, through this construction, fixing $\alpha$ furnishes exaclty one point $(f_\alpha^{-1}(\lambda_n), \alpha) \in \gamma_n$. Since $\lambda_n$ and $f_\alpha^{-1}(\lambda)$ are monotonically strictly decreasing respectively in $n$ and $\lambda$, it follows that $\mu_n = f_\alpha^{-1}(\lambda_n)$ is monotonically strictly increasing in $n$.
	
	\underline{Proof of Part \eqref{thm_1_part_4}}.
	
	By the small $\alpha$ limit properties of Proposition \ref{properties} in \eqref{lambda_a} and \eqref{lambda_c} and the fact that $\alpha > \alpha_0$ for $(\mu, \alpha) \in P^+_-$, one has that $\lim_{\alpha \to 0}\lambda_{(\mu, \alpha)}$ exists only for $\mu > 0$ and for this case one has $\lim_{\alpha \to 0}\lambda_{(\mu, \alpha)} = 2\pi\mu^{-1/2}$, which is the wavelength of the linear system. Matching the wavelength in this limit gives the eigenvalues $\mu^\mathrm{lin}_1 = 0$ of the linear problem. We note that $\mu^\mathrm{lin}_1$ is an eigenvalue of the linear problem with Neumann conditions at both ends of an interval. Yet it is not obtained as a limit of any spectral curve as $\gamma_1$ is missing from $\Phi_\mathrm{int}$ and in this case we actually do not have any solutions with no zeros at all.
	
	\end{proof}

	\section{Proof of Theorem \ref{thm_2}}\label{thm_2_proof}
	
	We prove Theorem \ref{thm_2} by considering a set of $d$ semi-infinite rays, each of which corresponds to one of the star edges. We equip each ray with the same boundary condition as its corresponding edge. Next we identify $d$ points, one on each ray, and require matching conditions on them such that combining the solutions of \eqref{line} on the rays yields a solution of \eqref{star}. Hence we start by describing a set of solutions of \eqref{line} on the ray.
	
	Let $\nu \ne 0$ and fix a Dirichlet or Neumann boundary condition on the endpoint of a ray, we denote
	\begin{align}
	\Phi_\mathrm{ray} := \{\phi \ne 0\text{ solves \eqref{line} on a ray, such that }\phi\text{ is periodic and attains zeros}\}.
	\end{align}
	
	We next follow the same path as for the line and the interval and parametrize solutions of \eqref{line} on a ray with points $(\mu, \alpha) \in P$, which we write as $\phi = \phi_{(\mu, \alpha)}$
	
	\begin{lem} The following holds.
	\begin{enumerate}
	\item Let $\nu \ne 0$. Given $\phi \in \Phi_\mathrm{ray}$, there is a unique value of $\mu \in \mathbb R$ such that $\phi$ is a solution of \eqref{line} on a ray with the given values $\mu, \nu$. This allows us to define the map
	\begin{align}
	\Lambda_\mathrm{ray}: \Phi_\mathrm{ray} \to P,\quad \Lambda_\mathrm{ray}: \phi \mapsto (\mu, ||\phi||_\infty),
	\end{align}
	which is onto and two to one since $\Lambda_\mathrm{ray}(\phi_1) = \Lambda_\mathrm{ray}(\phi_2)$ if and only if $\phi_1 = \zeta\phi_2$ for some fixed $\zeta = \pm1$.
	\item Every solution $\phi \in \Phi_\mathrm{ray}$ is a restriction of the solution $\widehat\phi = \Lambda_\mathrm{line}^{-1}\circ\Lambda_\mathrm{ray}(\phi)$ from the line to a ray, where $||\widehat\phi||_\infty = ||\phi||_\infty$. Note that $\Lambda_\mathrm{line}$ is not surjective and hence $\Lambda_\mathrm{line}^{-1}$ is not uniquely defined however, for the sake of the statement, any preimage of $\Lambda_\mathrm{line}$ can be chosen as the image of $\Lambda_\mathrm{line}^{-1}$.
	\end{enumerate}
	\end{lem}
	
	\begin{proof}
	\underline{Part (1).} The proofs of uniqueness of $\mu$ for a ray and that $\Lambda_\mathrm{ray}$ is onto are the same as that for the line, seen in the proof of Proposition \ref{sol_space}. The proof that $\Lambda_\mathrm{ray}$ is two to one is the same as that of the proof of Lemma \ref{lem_int} Part \eqref{lem_int_1}.
	
	\underline{Part (2).} The proof follows from the same argument as that of the proof of Lemma \ref{lem_int} Part \eqref{lem_int_2}.
	\end{proof}
	
	We have just established the existence and properties of $\Lambda_\mathrm{ray}$. The same was done for $\Lambda_\mathrm{star}$ in Lemma \ref{lem_star}, whose proof follows the same lines as Lemma \ref{lem_int}.
	
	\begin{proof}[Proof of Theorem \ref{thm_2} Part \eqref{thm_2_part_1}]

	Take a set of $d$ rays of the form $L_j = (-\infty, l_j]$ and take $\Phi_{\mathrm{ray}, j}$ to be the space of solutions of \eqref{line} on each such ray with the same boundary condition as the $j$-th edge of the star. Furthermore, define $\Lambda_{\mathrm{ray}, j}: \Phi_{\mathrm{ray}, j} \to P$ to act as $\Lambda_{\mathrm{ray}, j}: \phi_j \mapsto (\mu, \alpha_j = ||\phi_j||_\infty)$. Take $\phi_{(q_*)} \in \Phi_\mathrm{star}$ as in the statement of the theorem. Now for each $j$ we choose $\widetilde\Phi_{\mathrm{ray}, j} \subset \Phi_{\mathrm{ray}, j}$ such that $\Lambda_{\mathrm{ray}, j}\downharpoonright_{\widetilde\Phi_{\mathrm{ray}, j}}$ is one to one. Explicitly, this choice is made as follows. If the $j$-th ray has a Neumann condition, then
	\begin{align}
	\widetilde\Phi_{\mathrm{ray}, j} = \{\phi_j \in \Phi_{\mathrm{ray}, j}: \mathrm{sgn}(\phi_j(l_j)) = \mathrm{sgn}(\phi_{(q_*), j}(l_j))\}
	\end{align}
	and alternatively if the $j$-th ray has a Dirichlet condition, then
	\begin{align}
	\widetilde\Phi_{\mathrm{ray}, j} = \{\phi_j \in \Phi_{\mathrm{ray}, j}: \mathrm{sgn}(\partial_x\phi_j(l_j)) = \mathrm{sgn}(\partial_x\phi_{(q_*), j}(l_j))\}.
	\end{align}
	We further define the vector of constraints $N : Q \to \mathbb R^d$ to act as
	\begin{align}
	N_1(q) &:= \sum_{j = 1}^d\partial_x\phi_j(0),\label{N_1}\\
	N_j(q) &:= \phi_j(0) -\phi_1(0), \quad j = 2, \ldots, d,\label{N_2}
	\end{align}
	where $\phi_j \in \widetilde\Phi_{\mathrm{ray}, j}$ is uniquely obtained from $q = (\mu, \alpha_1, \ldots, d)$ by requiring $(\mu, \alpha_j) = \Lambda_{\mathrm{ray}, j}(\phi_j)$ for all $j$. We introduce the following useful notation. For each $\phi_{(q)} \in \Phi_\mathrm{star}$, denote by $\lambda_{(q)} \in \mathbb R^d$ the vector whose $j$-th component is the wavelength of the $\phi_j \in \widetilde\Phi_{\mathrm{ray}, j}$ corresponding to $j$-th edge of $\phi_{(q)}$, i.e. $\lambda_{(q), j} = \lambda_{(\mu, \alpha_j)}$.

	We observe that zeros of $N$ correspond to solutions of \eqref{star} in the following sense. For each $q = (\mu, \alpha_1, \ldots, \alpha_d)$ that satisfies $N(q) = 0$, take $\phi_j \in \widetilde\Phi_{\mathrm{ray}, j}$ such that $(\mu, \alpha_j) = \Lambda_{\mathrm{ray}, j}(\phi_j)$ for all $j$. Then there exists a solution of \eqref{star}, $\psi \in \Phi_\mathrm{star}$, such that $\psi_j = \phi_j\downharpoonright_{[0, l_j]}$ for all $j$. The validity of this solution indeed follows since the central vertex condition is satisfied as $N(q) = 0$. This motivates the application of Implicit Function Theorem to show the existence of a local spectral curve through $q_* \in Q$. In order to do so we need to assure the continuous differentiability of $N(q)$ in $q \in Q$, which is guaranteed by the following lemma, whose proof is postponed.
	
	\begin{lem}\label{cont_diff}
	For $\phi_{(\mu, \alpha)} \in \Phi_\mathrm{ray}$, $\phi_{(\mu, \alpha)}$ and $\partial_x\phi_{(\mu, \alpha)}$ are continuously differentiable in $\mu$ and $\alpha$.
	\end{lem}
	
	Now we aim to show that the Jacobian matrix $\partial_{\alpha_k}N_j(q_*)$ has a nonzero determinant. Let the linear operator $J : \mathbb R^d \to \mathbb R^d$ be defined by $J_{j, k} = \partial_{\alpha_k}N_j(q_*)$. We utilize the block decomposition
	\begin{align}
	J = \begin{bmatrix} A & B \\  C &  D \end{bmatrix},
	\end{align}
	where
	\begin{align}
	&J_{1, 1} = A := \partial_{\alpha_1}\partial_x\phi_{(q_*), 1}(0), \quad J_{1, j} = B_j := \partial_{\alpha_j}\partial_x\phi_{(q_*), j}(0),\\
	&J_{j, 1} = C_j := -\partial_{\alpha_1}\phi_{(q_*), 1}(0), \quad J_{j, j} = D_{j, j} := \partial_{\alpha_j}\phi_{(q_*), j}(0),
	\end{align}
	and where $j = 2, \ldots, d$ and $J_{j, k} = 0$ otherwise.

	We note that the diagonal components of $D$ are nonvanishing by the following argument. We now make use of Lemma \ref{ray_var}, which appears at the end of this section, and this is where the assumption that $\phi$ vanishes at the central vertex is being used. By \eqref{ray_var_3} there, one has that
	\begin{align}
	D_{j, j} = \partial_\alpha\phi_{(q_*), j}(0) = \zeta\xi\alpha_j(\mu +\nu\alpha_j^{2\sigma})^{1/2}\partial_\alpha\lambda_{(q), j} \downharpoonright_{q = q_*}.
	\end{align}
	The above RHS vanishes only if one of $\alpha_j$, $(\mu +\nu\alpha_j^{2\sigma})^{1/2}$, or $\partial_\alpha\lambda_{(q), j}$ vanishes. By the definition of $P$, neither $\alpha_j$ nor $(\mu +\nu\alpha_j^{2\sigma})^{1/2}$ can vanish. By Proposition \ref{properties} one has that $\partial_\alpha\lambda_{(q), j} \ne 0$ for all $\alpha > 0$, and therefore cannot vanish. Thus, the tentative assumption cannot hold and $D_{j, j} \ne 0$ for $j = 2, \ldots, d$.
	
	Since $\det D = \prod_{j = 2}^dD_{j, j} \ne 0$ one may find
	\begin{align}
	J = \begin{bmatrix} A & B \\  C &  D \end{bmatrix} = \begin{bmatrix} I & 0 \\ 0 & D \end{bmatrix} \begin{bmatrix} A & B \\ D^{-1} C & I \end{bmatrix}.
	\end{align}
	then
	\begin{align}
	&\det(A -BD^{-1}C) = \det J/\det D = A - \sum_{j = 2}^d B_jC_j/D_{j, j}\\
	&\qquad = \{\partial_{\alpha_1}\partial_x\phi_{(q), 1} -\sum_{j = 2}^d \partial_{\alpha_j}\partial_x\phi_{(q), j} [-\partial_{\alpha_1}\phi_{(q), 1}]/\partial_{\alpha_j}\phi_{(q), j}\}\downharpoonright_{x = 0, q = q_*}\\
	&\qquad = \partial_{\alpha_1}\phi_{(q), 1}\sum_{j = 1}^d \partial_{\alpha_j}\partial_x\phi_{(q), j}/\partial_{\alpha_j}\phi_{(q), j}\downharpoonright_{x = 0, q = q_*}.
	\end{align}
	Since $\partial_{\alpha_1}\phi_{(q), 1}\downharpoonright_{x = 0, q = q_*} \ne 0$ by the same argument that gives $D_{j, j} \ne 0$, if we denote
	\begin{align}
	S := \sum_{j = 1}^d \partial_{\alpha_j}\partial_x\phi_{(q), j}/\partial_{\alpha_j}\phi_{(q), j}\downharpoonright_{x = 0, q = q_*}
	\end{align}
	then $\det J = 0$ if and only if $S = 0$. We will now show that $S \ne 0$. We define
	\begin{align}
	t_j := [\mu +(\sigma +1)\nu\alpha_j^{2\sigma}]/(\mu +\nu\alpha_j^{2\sigma})
	\end{align}
	and note that due to \eqref{ray_var_3} and \eqref{ray_var_4} of Lemma \ref{ray_var} one may write
	\begin{align}
	S = \sum_{j = 1}^d t_j/(\xi_j\alpha_j\partial_\alpha\lambda_{(q), j})\downharpoonright_{q = q_*}.
	\end{align}
	We recall from Proposition \ref{properties} that one has $\partial_\alpha\lambda_{(q), j} < 0$ for $\nu > 0$ and $\partial_\alpha\lambda_{(q), j} > 0$ for $\nu < 0$. Therefore $S \ne 0$ if $\mathrm{sgn}(t_j)$ is nonvanishing and constant for all $j$.
	
	First consider $(\mu, \alpha_j) \in P^+_+$ for all $j$; clearly since $\nu, \mu, \alpha > 0$ it must be true that $t_j > 0$. Now consider $(\mu, \alpha_j) \in P^+_-$ for all $j$; one has
	\begin{align}
	t_j = [\mu +(\sigma +1)\nu\alpha_j^{2\sigma}]/(\mu +\nu\alpha_j^{2\sigma}) = [(\sigma +1)|\nu/\mu|\alpha_j^{2\sigma} -1]/(|\nu/\mu|\alpha_j^{2\sigma} -1),
	\end{align}
	and since $\alpha_j \in (\alpha_0, \infty)$ and $\alpha_0 = |\mu/\nu|^{1/2\sigma}$ we get that $t_j > 2 > 0$. Lastly consider $(\mu, \alpha_j) \in P^-$ for all $j$; one has
	\begin{align}
	t_j = [\mu +(\sigma +1)\nu\alpha_j^{2\sigma}]/(\mu +\nu\alpha_j^{2\sigma}) = [1 -(\sigma +1)|\nu/\mu|\alpha_j^{2\sigma}]/(1 -|\nu/\mu|\alpha_j^{2\sigma}),
	\end{align}
	and since $\alpha_j \in (0, \alpha_\mathrm{c})$ and $\alpha_\mathrm{c} = |\mu/(\sigma +1)\nu|^{1/2\sigma}$ we get that $t_j > 0$. Therefore it must be true that $t_j > 0$ for all $j$. Thus $\det S \ne 0$ and in turn $\det J \ne 0$.
	
	We have ensured that requirements for Implicit Function Theorem are assured. Therefore its application, as stated in the beginning of the proof, guarantees the existence of a local spectral curve through $q_*$.
	\end{proof}
	
	\begin{proof}[Proof of Theorem \ref{thm_2} Part \eqref{thm_2_part_2}]
	
	Along a local spectral curve the solution always satisfies the matching conditions and is continuous. Although the shape of the oscillations can change slightly, the most dramatic change occurs at the central vertex. If $\phi_{(q_*)}(0) = 0$, then $\mu \mapsto \mu +\delta\mu$ can yield only three results: $\phi_{(q)}(0) = 0$, $\phi_{(q)}(0) > 0$, or $\phi_{(q)}(0) < 0$. One can always find a $\delta\mu$ small enough so that the variation does not move through more than one of these cases and in this sense one must take $q$ to be appropriately close to $q_*$. If the central value remains zero, then no change in nodal count can occur, hence the need for the factor of $\mathrm{sgn}^2(\phi_{(q)})$. If the central value increases from zero, then the zero at the center vanishes and a zero forms on each edge for which $\partial_x\phi_{(q_*), j}(0) < 0$. The contribution to the nodal count change is then a sum of a term $-1$ for the vanishing of the zero at the center and a term $2^{-1}[1 -\mathrm{sgn}(\phi_{(q)}\partial_x\phi_{(q_*), j})]\downharpoonright_{x = 0}$ for each edge. If the central value decreases from zero, then the zero at the center vanishes and a zero forms on each edge for which $\partial_x\phi_{(q_*), j}(0) > 0$. The contribution to the nodal count change is the same as before, i.e. a sum of a term $-1$ for the vanishing of the zero at the center and a term $2^{-1}[1 -\mathrm{sgn}(\phi_{(q)}\partial_x\phi_{(q_*), j})]\downharpoonright_{x = 0}$ for each edge. Combining these contributions and considerations for each case, one has
	\begin{align}
	&Z(\phi_{(q)}) -Z(\phi_{(q_*)}) = \mathrm{sgn}^2(\phi_{(q)})\left\{-1 +2^{-1}\sum_{j = 1}^d\left[1 -\mathrm{sgn}(\phi_{(q)}\partial_x\phi_{(q_*), j})\right]\right\}\downharpoonright_{x = 0}\\
	&\qquad = \mathrm{sgn}^2(\phi_{(q)})\left[-1 +2^{-1}d -2^{-1}\mathrm{sgn}(\phi_{(q)})\sum_{j = 1}^d\mathrm{sgn}(\partial_x\phi_{(q_*), j})\right]\downharpoonright_{x = 0}.
	\end{align}
	
	\end{proof}
	
	\begin{proof}[Proof of Theorem \ref{thm_2} Part \eqref{thm_2_part_3}]
	We are at liberty to construct the desired $\phi_{(q_*)} \in \Phi_\mathrm{star}$ so that it automatically satisfies the desired properties. For only Dirichlet conditions on exterior vertices and $\phi(0) = 0$ we must have $l_j = n_j\lambda_{(q), j}/2$ for some $n_j \in \mathbb N$ and all $j$. Since for $\mu = 0$ one has from equation \eqref{zero} of Lemma \ref{mu_zero}, which appears at the end of this section, that
	\begin{align}
	\lambda_{(0, \alpha_j)} = 4c_1\nu^{-1/2}\alpha_j^{-\sigma},
	\end{align}
	where $c_1$ is given in \eqref{c_1} of Lemma \ref{mu_zero}, it must then be the case that
	\begin{align}
	\alpha_j = (2c_1\nu^{-1/2}n_j/l_j)^{1/\sigma}.\label{alpha}
	\end{align}
	Since $\phi_{(q_*), j}(0) = 0$ for all $j$, it is the case that $0 = N_j(q_*)$, as required by \eqref{N_1} and \eqref{N_2}, is automatically satisfied for $j = 2, \ldots, d$. We now check what is required for the matching condition at the central vertex. We take
	\begin{align}\label{signs}
	\zeta_j = \mathrm{sgn}(\partial_x\phi_{(q_*), j}(0))
	\end{align}
	for all $j$. We recall from the definition of $N_1$ in \eqref{N_1}, the expression of $\partial_x\phi$ in \eqref{energy}, and $\mu_* = 0$ in turn
	\begin{align}
	0 &= N_1(q_*) = \sum_{j = 1}^d \partial_x\phi_{(q_*), j}(0) = \sum_{j = 1}^d \zeta_j[\mu\alpha_j^2 +\nu\alpha_j^{2(\sigma +1)}]^{1/2} \downharpoonright_{q = q_*}\label{sub}\\
	&= \sum_{j = 1}^d \zeta_j\nu^{1/2}\alpha_j^{\sigma +1} \downharpoonright_{q = q_*} = \sum_{j = 1}^d \zeta_j\nu^{1/2}(2c_1\nu^{-1/2}n_j/l_j)^{1 +1/\sigma}\\
	&\Leftrightarrow\quad 0 = \sum_{j = 1}^d \zeta_j(n_j/l_j)^{1 +1/\sigma},
	\end{align}
	which is automatically satisfied by assumption for appropriately chosen $\zeta_j, n_j$. Since we have specified $\alpha_j$ and $\mathrm{sgn}(\partial_x\phi_{(q_*), j}(0))$, for all j, the solution $\phi_{(q_*)}$ has been uniquely determined. The nodal count $Z(\phi_{(q_*)}) = 1 -d +\sum_{j = 1}^dn_j$ follows by inspection of the hereto constructed solution $\phi_{(q_*)}$, which completes the proof.
	\end{proof}
	
	\begin{rem}
	As mentioned after the statement of the theorem, the proof above generalizes to the case where some (or all) of the Dirichlet boundary conditions are replaced by Neumann ones. This may be obtained by replacing, in the assumed condition and proof, $n_j$ with $n_j -1/2$ for each edge that possesses a Neumann boundary condition at its end. Nevertheless, the nodal count expression is unchanged.
	\end{rem}
	
	\begin{proof}[Proof of Theorem \ref{thm_2} Part \eqref{thm_2_part_4}]
	
	We first combine the conclusions of Parts \eqref{thm_2_part_1} and \eqref{thm_2_part_3} of the Theorem and get the existence of a local spectral curve $\gamma(\mu) = (\mu, \gamma_1(\mu), \ldots, \gamma_d(\mu))$ through $q_*$ given in \eqref{thm_2_part_3}. We now tentatively assume for the sake of contradiction that
	\begin{align}
	\frac{\mathrm d}{\mathrm d\mu}\phi_{(\gamma(\mu_*)), j}(0) = 0,
	\end{align}
	where the complete derivative is taken along the spectral curve $\gamma$ mentioned above. From $\mathrm d\phi_{(\gamma(\mu_*)), j}/\mathrm d\mu = 0$ we also conclude that $\mathrm d\lambda_{(\gamma(\mu_*)), j}/\mathrm d\mu = 0$, for which one must have
	\begin{align}
	&0 = \frac{\mathrm d}{\mathrm d\mu}\lambda_{(q_*), j} = \partial_\mu\lambda_{(q_*), j} +\partial_{\alpha_j}\lambda_{(q_*), j}\frac{\mathrm d}{\mathrm d\mu}\gamma_j(\mu_*)\\
	\Rightarrow\quad &\frac{\mathrm d}{\mathrm d\mu}\gamma_j(\mu_*) = -\partial_\mu\lambda_{(q_*), j}/\partial_{\alpha_j}\lambda_{(q_*), j},
	\end{align}
	for all $j = 1, \ldots, d$ and then for $\mu_* = 0$ one has by Lemma \ref{mu_zero} that
	\begin{align}
	\frac{\mathrm d}{\mathrm d\mu}\gamma_j(\mu_*) &= -(-2c_2\nu^{-3/2}\alpha_j^{-3\sigma})/[-4\sigma c_1\nu^{-1/2}\alpha_j^{-(\sigma +1)}] \downharpoonright_{q = q_*}\\
	&= -2^{-1}\sigma^{-1}c_3\nu^{-1}\alpha_j^{-(2\sigma -1)} \downharpoonright_{q = q_*},
	\end{align}
	where $c_1, c_2, c_3$ are given respectively by \eqref{c_1}, \eqref{c_2}, \eqref{c_3}, of Lemma \ref{mu_zero}.
	
	We recall from \eqref{ray_var_2} and \eqref{ray_var_4} of Lemma \ref{ray_var} that for $\mu_* = 0$ one has
	\begin{align}
	\partial_\mu\partial_x\phi_{(q_*), j}(0) &= 2^{-1}\zeta_j\nu^{-1/2}\alpha_j^{1 -\sigma} \downharpoonright_{q = q_*}\\
	\partial_\alpha\partial_x\phi_{(q_*), j}(0) &= (\sigma +1)\zeta_j\nu^{1/2}\alpha_j^\sigma \downharpoonright_{q = q_*},
	\end{align}
	and since $N_1(q) = \sum_{j = 1}^d\partial_x\phi_{(q), j}(0)$, as defined by \eqref{N_1}, one has for $\mu_* = 0$ that
	\begin{align}
	&\frac{\mathrm d}{\mathrm d\mu}N_1(\gamma(\mu_*)) = \left[\partial_\mu N_1(q_*) +\sum_{j = 1}^d\partial_{\alpha_j}N_1(q_*)\frac{\mathrm d}{\mathrm d\mu}\gamma_j(\mu_*)\right]\\
	&\qquad = \sum_{j = 1}^d\left[\partial_\mu\partial_x\phi_{(q_*), j}(0) +\partial_{\alpha_j}\partial_x\phi_{(q_*), j}(0)\frac{\mathrm d}{\mathrm d\mu}\gamma_j(\mu_*)\right]\\
	&\qquad = \sum_{j = 1}^d\{2^{-1}\zeta_j\nu^{-1/2}\alpha_j^{1 -\sigma} +[(\sigma +1)\zeta_j\nu^{1/2}\alpha_j^\sigma][-2^{-1}\sigma^{-1}c_3\nu^{-1}\alpha_j^{-(2\sigma -1)}]\} \downharpoonright_{q = q_*}\\
	&\qquad = \sum_{j = 1}^d\zeta_j[2^{-1}\nu^{-1/2}\alpha_j^{1 -\sigma} -2^{-1}\sigma^{-1}(\sigma +1)c_3\nu^{-1/2}\alpha_j^{1 -\sigma}] \downharpoonright_{q = q_*}.
	\end{align}
	Then by \eqref{alpha} of the previous part of the proof
	\begin{align}
	&\frac{\mathrm d}{\mathrm d\mu}N_1(\gamma(\mu_*)) = 2^{-1}\nu^{-1/2}[1 -(1 +1/\sigma)c_3]\sum_{j = 1}^d\zeta_j(2c_1\nu^{-1/2}n_j/l_j)^{-1 +1/\sigma}\\
	&\qquad = 2^{-1}\nu^{-1/2}[1 -(1 +1/\sigma)c_3](2c_1\nu^{-1/2})^{1/\sigma -1}\sum_{j = 1}^d\zeta_j(n_j/l_j)^{-1 +1/\sigma} \ne 0,
	\end{align}
	which contradicts $\gamma$ being a local spectral curve passing through $q_*$. By this contradiction, it must be the case that our tentative assumption is false and therefore $\mathrm d\phi_{(\gamma_1(\mu_*)), j}(0)/\mathrm d\mu \ne 0$, for all $j$. As a consequence hereof, taken with the fact that $\phi_{(q)}(0)$ is continuous in $q$ thanks to Lemma \ref{cont_diff}, it must be the case that through this $q_*$ passes a local spectral curve $\gamma$ such that for all $q_+, q_- \in \gamma$ sufficiently close to $q_*$, where $\mu_+ > 0 > \mu_-$, one has $\mathrm{sgn}(\phi_{(q_+)}) = -\mathrm{sgn}(\phi_{(q_-)})$. Then using \eqref{nodal_count_change} and \eqref{signs}, we find
	\begin{align}
	&Z(\phi_{(q_+)}) -Z(\phi_{(q_-)}) = [Z(\phi_{(q_+)}) -Z(\phi_{(q_*)})] -[Z(\phi_{(q_-)}) -Z(\phi_{(q_*)})]\\
	&\qquad = \left\{\mathrm{sgn}^2(\phi_{(q_+)})\left[-1 +2^{-1}d -2^{-1}\mathrm{sgn}(\phi_{(q_+)})\sum_{j = 1}^d\mathrm{sgn}(\partial_x\phi_{(q_*), j})\right]\downharpoonright_{x = 0}\right\}\\
	&\qquad\qquad -\left\{\mathrm{sgn}^2(\phi_{(q_-)})\left[-1 +2^{-1}d -2^{-1}\mathrm{sgn}(\phi_{(q_-)})\sum_{j = 1}^d\mathrm{sgn}(\partial_x\phi_{(q_*), j})\right]\downharpoonright_{x = 0}\right\}\\
	&\qquad = \mathrm{sgn}(\phi_{(q_-)})\sum_{j = 1}^d\zeta_j,
	\end{align}
	which gives $|Z(\phi_{(q_+)}) -Z(\phi_{(q_-)})| = |\sum_{j = 1}^d\zeta_j|$ and completes the proof.
	
	\end{proof}
	
	We end this section by proving Lemma \ref{cont_diff}, which was stated earlier in the proof of part \eqref{thm_2_part_1} of Theorem \ref{thm_2}, and stating and proving Lemmata \ref{ray_var} and \ref{mu_zero}, which were referenced above.
	
	\begin{proof}[Proof of Lemma \ref{cont_diff}]
	We recall from \eqref{integral} of the proof Proposition \ref{sol_space} that the solutions of \eqref{line} may be formed as the inverse functions of
	\begin{align}
	&\widehat x(\widehat\phi_0, \widehat\phi) = x_0 +\zeta\int_{\widehat\phi_0}^{\widehat\phi}\mathrm dw\ [h -\mu w^2 -\nu w^{2(\sigma +1)}]^{-1/2}\\
	&\qquad = x_0 +\zeta\int_{\widehat\phi_0}^{\widehat\phi}\mathrm dw\ \{\mu(\alpha^2 -w^2) +\nu[\alpha^{2(\sigma +1)} -w^{2(\sigma +1)}]\}^{-1/2}.
	\end{align}
	For $\phi = \phi_{(p)} \in \Phi_\mathrm{ray}$ one may write part of the inverse solution on each quarter wavelength, starting at a zero of $\phi$, as
	\begin{align}
	\widehat x(\widehat\phi) = \zeta\int_0^{\widehat\phi}\mathrm dw\ \{\mu(\alpha^2 -w^2) +\nu[\alpha^{2(\sigma +1)} -w^{2(\sigma +1)}]\}^{-1/2}.
	\end{align}
	Since the full solution can be composed of appropriately gluing together partial solutions, it is sufficient to confirm that each partial solution is continuously differentiable in $\mu$ and $\alpha$. Let
	\begin{align}
	F(x, \mu, \alpha, \widehat\phi) := x -\zeta\int_0^{\widehat\phi}\mathrm dw\ \{\mu(\alpha^2 -w^2) +\nu[\alpha^{2(\sigma +1)} -w^{2(\sigma +1)}]\}^{-1/2}
	\end{align}
	so that the formula for the inverse function of the partial solution may be expressed as $F(x, \mu, \alpha, \widehat\phi) = 0$.
	
	It is easy to see that $F$ is continuously differentiable in $x$ and $\widehat\phi$. $F$ can be shown to be continuously differentiable in in $\mu$ and $\alpha$ by arguments similar to those used to study $\partial_\mu\lambda$ and $\partial_\alpha\lambda$ in the proof of the monotonicity properties of Proposition \ref{properties}. Furthermore, since $\partial_{\widehat\phi}F(x, \mu, \alpha, \widehat\phi) \ne 0$ for all $x, \mu, \alpha$, Implicit Function Theorem gives that there must exist a $\phi_{(\mu, \alpha)}(x)$ that is continuously differentiable in $\mu$, $\alpha$, and $x$ in an open neighborhood of any point $(\mu, \alpha, x) \in P \times \mathbb R$. This continuous differentiability also holds in the same manner for $\partial_x\phi_{(\mu, \alpha)}(x)$ since by \eqref{energy} one has
	\begin{align}
	\partial_x\phi_{(\mu, \alpha)}(x) = \zeta\{\mu[\alpha^2 -\phi_{(\mu, \alpha)}^2(x)] +\nu[\alpha^{2(\sigma +1)} -\phi_{(\mu, \alpha)}^{2(\sigma +1)}(x)]\}^{1/2},\label{partial_phi}
	\end{align}
	which completes the proof.
	\end{proof}
	
	\begin{rem}
	The lemma above could also be proven, for the case of $\sigma = 1$, by referring to Jacobi elliptic functions, which are the explicit solutions of \eqref{line} in this case.
	\end{rem}
	
	\begin{lem}\label{ray_var}
	Let $\phi_{(\mu_*, \alpha_*)} \in \Phi_\mathrm{ray}$ where the ray has the form $(-\infty, l]$. If $\phi_{(\mu_*, \alpha_*)}(0) = 0$ then
	\begin{align}
	\partial_\mu\phi_{(\mu_*, \alpha_*)}(0) &= \zeta\xi\alpha(\mu +\nu\alpha^{2\sigma})^{1/2}\partial_\mu\lambda_{(\mu, \alpha)}\downharpoonright_{(\mu, \alpha) = (\mu_*, \alpha_*)},\label{ray_var_1}\\
	\partial_\mu\partial_x\phi_{(\mu_*, \alpha_*)}(0) &= 2^{-1}\zeta(\mu +\nu\alpha^{2\sigma})^{-1/2}\alpha \downharpoonright_{(\mu, \alpha) = (\mu_*, \alpha_*)},\label{ray_var_2}\\
	\partial_\alpha\phi_{(\mu_*, \alpha_*)}(0) &= \zeta\xi\alpha(\mu +\nu\alpha^{2\sigma})^{1/2}\partial_\alpha\lambda_{(\mu, \alpha)} \downharpoonright_{(\mu, \alpha) = (\mu_*, \alpha_*)},\label{ray_var_3}\\
	\partial_\alpha\partial_x\phi_{(\mu_*, \alpha_*)}(0) &= \zeta(\mu +\nu\alpha^{2\sigma})^{-1/2}[\mu +(\sigma +1)\nu\alpha^{2\sigma}] \downharpoonright_{(\mu, \alpha) = (\mu_*, \alpha_*)}\label{ray_var_4},
	\end{align}
	where $\zeta = \pm1$ and $l = \xi\lambda_{(\mu_*, \alpha_*)}$ for some fixed $\xi \in \mathbb N/4$.
	\end{lem}
	
	\begin{proof}
	Let $\widehat x_{(\mu, \alpha)}$ satisfy $\widehat x_{(\mu_*, \alpha_*)} = 0$ and $\phi_{(\mu, \alpha)}(\widehat x_{(\mu, \alpha)}) = 0$ for all $(\mu, \alpha)$. Namely $\widehat x_{(\mu, \alpha)}$ denotes the position of some nodal point of the solution $\phi_{(\mu, \alpha)}$. Therefore the variation of $\widehat x_{(\mu, \alpha)}$ with respect to $\mu$ must take the form $\partial_\mu\widehat x_{(\mu, \alpha)} = -\xi\partial_\mu\lambda_{(\mu, \alpha)}$ for all $(\mu, \alpha)$ and for some fixed $\xi \in \mathbb N/4$. Therefore
	\begin{align}
	&0 = \phi_{(\mu, \alpha)}(\widehat x_{(\mu, \alpha)})\\
	\Rightarrow\quad &0 = \frac{\mathrm d}{\mathrm d\mu}\phi_{(\mu, \alpha)}(\widehat x_{(\mu, \alpha)}) = \partial_\mu\phi_{(\mu, \alpha)}(\widehat x_{(\mu, \alpha)}) +\partial_x\phi_{(\mu, \alpha)}(\widehat x_{(\mu, \alpha)})\partial_\mu\widehat x_{(\mu, \alpha)},\\
	\Rightarrow\quad &\partial_\mu\phi_{(\mu_*, \alpha_*)}(0) = \xi\partial_x\phi_{(\mu_*, \alpha_*)}(0)\partial_\mu\lambda_{(\mu_*, \alpha_*)} = \zeta\xi\alpha(\mu +\nu\alpha^{2\sigma})^{1/2}\partial_\mu\lambda_{(\mu, \alpha)}\downharpoonright_{(\mu, \alpha) = (\mu_*, \alpha_*)},\label{eval_1}
	\end{align}
	where we have used
	\begin{align}
	\partial_x\phi_{(\mu_*, \alpha_*)}(0) &=  \zeta\alpha(\mu +\nu\alpha^{2\sigma})^{1/2} \downharpoonright_{(\mu, \alpha) = (\mu_*, \alpha_*)},
	\end{align}
	which holds thanks to \eqref{partial_phi} and $\phi_{(\mu_*, \alpha_*)}(\widehat x_{(\mu_*, \alpha_*)}) = 0$. Furthermore, by differentiating the expression of $\partial_x\phi$ in \eqref{partial_phi} with respect to $\mu$, one has
	\begin{align}
	&[\partial_\mu\partial_x\phi_{(\mu, \alpha)}(\widehat x_{(\mu, \alpha)})] \downharpoonright_{(\mu, \alpha) = (\mu_*, \alpha_*)}\\
	&\qquad = 2^{-1}\zeta\{\mu[\alpha^2 -\phi_{(\mu, \alpha)}^2(\widehat x_{(\mu, \alpha)})] +\nu[\alpha^{2(\sigma +1)} -\phi_{(\mu, \alpha)}^{2(\sigma +1)}(\widehat x_{(\mu, \alpha)})]\}^{-1/2}\\
	&\qquad\qquad \times \{[\alpha^2 -\phi_{(\mu, \alpha)}^2(\widehat x_{(\mu, \alpha)})] -2\phi_{(\mu, \alpha)}(\widehat x_{(\mu, \alpha)})\partial_\mu[\phi_{(\mu, \alpha)}(\widehat x_{(\mu, \alpha)})]\\
	&\qquad\qquad -2(\sigma +1)\nu\phi_{(\mu, \alpha)}^{2\sigma +1}(\widehat x_{(\mu, \alpha)})\partial_\mu[\phi_{(\mu, \alpha)}(\widehat x_{(\mu, \alpha)})]\} \downharpoonright_{(\mu, \alpha) = (\mu_*, \alpha_*)}\\
	&\qquad = 2^{-1}\zeta[\mu\alpha^2 +\nu\alpha^{2(\sigma +1)}]^{-1/2}\alpha^2 \downharpoonright_{(\mu, \alpha_) = (\mu_*, \alpha_*)}\\
	&\qquad = 2^{-1}\zeta(\mu +\nu\alpha^{2\sigma})^{-1/2}\alpha \downharpoonright_{(\mu, \alpha) = (\mu_*, \alpha_*)}\label{eval_2}.
	\end{align}
	
	The variation of $\widehat x_{(\mu, \alpha)}$ with respect to $\alpha$ must take the form $\partial_\alpha\widehat x_{(\mu, \alpha)} = -\xi\partial_\alpha\lambda_{(\mu, \alpha)}$ for all $(\mu, \alpha)$. Therefore we find similarly to \eqref{eval_1} above that
	\begin{align}
	\partial_\alpha\phi_{(\mu_*, \alpha_*)}(0) = \xi\partial_x\phi_{(\mu_*, \alpha_*)}(0)\partial_\alpha\lambda_{(\mu_*, \alpha_*)} = \zeta\xi\alpha(\mu +\nu\alpha^{2\sigma})^{1/2}\partial_\alpha\lambda_{(\mu, \alpha)}\downharpoonright_{(\mu, \alpha) = (\mu_*, \alpha_*)}.
	\end{align}
	Furthermore, by differentiating the expression of $\partial_x\phi$ in \eqref{partial_phi} with respect to $\alpha$, one has, by a similar calculation to that of \eqref{eval_2} above, that
	\begin{align}
	&\partial_\alpha\partial_x\phi_{(\mu_*, \alpha_*)}(0) = 2^{-1}\zeta[\mu\alpha^2 +\nu\alpha^{2(\sigma +1)}]^{-1/2}[2\mu\alpha +2(\sigma +1)\nu\alpha^{2\sigma +1}] \downharpoonright_{(\mu, \alpha) = (\mu_*, \alpha_*)}\\
	&\qquad = \zeta(\mu +\nu\alpha^{2\sigma})^{-1/2}[\mu +(\sigma +1)\nu\alpha^{2\sigma}]\downharpoonright_{(\mu, \alpha) = (\mu_*, \alpha_*)}.
	\end{align}
	\end{proof}
	
	\begin{lem}\label{mu_zero}
	For $\mu = 0$ one has
	\begin{align}
	\lambda_{(0, \alpha)} &= 4c_1\nu^{-1/2}\alpha^{-\sigma},\label{zero}\\
	\partial_\alpha\lambda_{(0, \alpha)} &= -4\sigma c_1\nu^{-1/2}\alpha^{-(\sigma +1)},\label{one}\\
	\partial_\mu\lambda_{(0, \alpha)} &= -2c_2\nu^{-3/2}\alpha^{-3\sigma},\label{two}
	\end{align}
	where
	\begin{align}
	c_1 &:= \int_0^1 \mathrm dw\ (1 - w^{2(\sigma +1)})^{-1/2} > 0,\label{c_1}\\
	c_2 &:= \int_0^1 \mathrm dw\ (1 - w^{2(\sigma +1)})^{-3/2}(1 -w^2) > 0,\label{c_2}\\
	c_3 &:= c_2/c_1 > 0.\label{c_3}
	\end{align}
	\end{lem}
	
	\begin{proof}
	Due to \eqref{lambda} of Proposition \ref{properties} one has
	\begin{align}
	\lambda_{(0, \alpha)} &= 4\int_0^1 \mathrm{d}w\ [\nu\alpha^{2\sigma}(1 - w^{2(\sigma +1)})]^{-1/2} = 4\nu^{-1/2}\alpha^{-\sigma}\int_0^1 \mathrm{d}w\ (1 - w^{2(\sigma +1)})^{-1/2}\\
	&= 4c_1\nu^{-1/2}\alpha^{-\sigma}.
	\end{align}
	Differentiating $\lambda_{(\mu, \alpha)}$, as given in \eqref{lambda}, one has
	\begin{align}
	\partial_\alpha\lambda_{(\mu, \alpha)}/4 &= -2^{-1}\int_0^1 \mathrm{d}w\ [\mu(1 -w^2) +\nu\alpha^{2\sigma}(1 - w^{2(\sigma +1)})]^{-3/2}2\sigma\nu\alpha^{2\sigma -1}(1 -w^{2(\sigma +1)})\\
	\partial_\alpha\lambda_{(0, \alpha)} &= -4\sigma\nu^{-1/2}\alpha^{-(\sigma +1)}\int_0^1 \mathrm{d}w\ (1 - w^{2(\sigma +1)})^{-3/2}(1 -w^{2(\sigma +1)})\\
	&= -4\sigma\nu^{-1/2}\alpha^{-(\sigma +1)}\int_0^1 \mathrm{d}w\ (1 - w^{2(\sigma +1)})^{-1/2} = -4\sigma c_1\nu^{-1/2}\alpha^{-(\sigma +1)},
	\end{align}
	and
	\begin{align}
	\partial_\mu\lambda_{(\mu, \alpha)}/4 &= -2^{-1}\int_0^1 \mathrm{d}w\ [\mu(1 -w^2) +\nu\alpha^{2\sigma}(1 - w^{2(\sigma +1)})]^{-3/2}(1 -w^2)\\
	\partial_\mu\lambda_{(0, \alpha)} &= -2\nu^{-3/2}\alpha^{-3\sigma}\int_0^1 \mathrm dw\ (1 - w^{2(\sigma +1)})^{-3/2}(1 -w^2) = -2c_2\nu^{-3/2}\alpha^{-3\sigma}.
	\end{align}
	\end{proof}
	
	\ \\
	\thanks{We are indebted to Sven Gnutzmann for inspiring this work. We thank Sebastian Egger his careful reading and useful direction. We thank Lior Alon and Michael Bersudsky for their helpful comments.}

	\end{document}